\documentclass{llncs}
\usepackage[utf8]{inputenc}
\usepackage{amsmath}
\usepackage{amssymb}
\usepackage{mathtools}
\usepackage{amsfonts}
\usepackage{xspace}
\usepackage{color,graphicx,grffile}
\usepackage[figure,tworuled,linesnumbered]{algorithm2e}  
\usepackage{multirow}
\usepackage{soul}
\usepackage{tipa}
\usepackage{enumitem}
\usepackage{grffile}
\usepackage{svg}

\usepackage[text={15.4cm,19.5cm},centering]{geometry}

\newcommand{\GS}{\mathit{G_S}} 
\newcommand{\GT}{\mathit{G_T}} 
\newcommand{\GH}{\mathit{G_H}} 

\newcommand{\fsync}{{\sc FSync}\xspace}
\newcommand{\ssync}{{\sc SSync}\xspace}
\newcommand{\async}{{\sc Async}\xspace}
\newcommand{\sasync}{{\sc SAsync}\xspace}

\newcommand{\Wait}{{\tt Wait}\xspace}
\newcommand{\Look}{{\tt Look}\xspace}
\newcommand{\Compute}{{\tt Compute}\xspace}
\newcommand{\Move}{{\tt Move}\xspace}
\newcommand{\LCM}{{\tt LCM}\xspace}

\newcommand{\nil}{\mathit{nil}}

\newcommand{\Ex}{\mathbb{E}}  

\newcommand{\I}{\mathcal{I}} 

\renewcommand{\max}{\mathit{max}}

\newcommand{\Aut}[1]{\mbox{Aut}({#1})}

\newcommand{\Prop}{ \mathsf{Prop} }

\newcommand{\mbr}{\mathit{mbr}}
\newcommand{\bp}{\mathit{bp}}
\newcommand{\mbp}{\mathit{mbp}}
\newcommand{\w}{\mathit{w}}
\newcommand{\h}{\mathit{h}}
\newcommand{\PP}{\mathit{P^*}}

\newcommand{\LSS}{\textit{LSS}\xspace}
\newcommand{\LSF}{\mathit{\ell^F_{f_1}}}

\newcommand{\apf}{\mathit{APF}\xspace}

\newcommand{\RS}{\mathit{RS}}
\newcommand{\PPF}{\mathit{PPF}}
\newcommand{\Fin}{\mathit{Fin}}
\newcommand{\Term}{\mathit{Term}}

\newcommand{\A}{\mathcal{A}} 
\newcommand{\Aform}{\mathcal{A}_{\mathit{form}}}

\definecolor{linecomment}{rgb}{0.95, 0.1, 0.1}
\definecolor{comment}{rgb}{0.1, 0.1, 0.95}

\newcommand{\pre}{\mathtt{pre}}

\newcommand{\false}{\mathtt{false}}

\newcommand{\vguno}{\mathtt{g1}}
\newcommand{\vgenne}{\mathtt{gn}}
\newcommand{\vpfuno}{\mathtt{pf1}}

\newcommand{\vpfenne}{\mathtt{pfn}}

\newcommand{\vhrenne}{\mathtt{hrn}}
\newcommand{\vs}{\mathtt{s}}
\newcommand{\vqfuno}{\mathtt{qf1}}
\newcommand{\vdruno}{\mathtt{dr1}}
\newcommand{\vdrunop}{\mathtt{dr1'}}
\newcommand{\vhpuno}{\mathtt{hp'}}
\newcommand{\vhpdue}{\mathtt{hp''}}
\newcommand{\vrpeffe}{\mathtt{rpf}}

\newcommand{\predue}{ \vguno }
\newcommand{\pretre}{ \vguno \wedge \vhpdue \wedge \vdruno}
\newcommand{\prequattro}{ \vguno \wedge \vdruno \wedge \vgenne \wedge \vrpeffe}
\newcommand{\precinque}{ \vguno \wedge \vhpuno \wedge \vdruno \wedge \vhrenne \wedge \vpfenne }
\newcommand{\presei}{ \vguno \wedge \vdrunop \wedge \vpfuno }
\newcommand{\presette}{ \vqfuno }

\begin{document}
	
\title{Arbitrary Pattern Formation on Infinite Regular Tessellation Graphs%
\thanks{The work has been supported in part by the 
Italian National Group for Scientific Computation (GNCS-INdAM).}
}

\author{%
Serafino Cicerone\inst{1}, 
Alessia Di Fonso\inst{1}, 
Gabriele Di Stefano\inst{1}, 
Alfredo Navarra\inst{2}
}


\institute{Dipartimento di Ingegneria e Scienze dell'Informazione e Matematica,
        Università  degli Studi dell'Aquila, I-67100 
        L'Aquila, Italy.
\email{serafino.cicerone@univaq.it},
\email{alessia.difonso@graduate.univaq.it},
\email{gabriele.distefano@univaq.it}
\and 
Dipartimento di Matematica e Informatica,
        Università degli Studi di Perugia I-06123 
        Perugia, Italy.
\email{alfredo.navarra@unipg.it} 
}

\maketitle


\begin{abstract}
Given a set $R$ of robots, each one located at different vertices of an infinite regular tessellation graph, we aim to explore the \emph{Arbitrary Pattern Formation} ($\apf$) problem. 
Given a multiset $F$ of grid vertices such that $|R|=|F|$, $\apf$ asks for a distributed algorithm that moves robots so as to reach a configuration similar to $F$. Similarity means that robots must be disposed as $F$ regardless of translations, rotations, reflections. 

So far, as possible graph discretizing the Euclidean plane only the standard square grid has been considered in the context of the classical \emph{Look-Compute-Move} model. However, it is natural to consider also the other regular tessellation graphs, that are triangular and hexagonal grids. 

We provide a resolution algorithm for $\apf$ when the initial configuration is asymmetric and the considered topology is any regular tessellation graph. 
\end{abstract}


\section{Introduction}\label{Sec:Introduction}

In this paper, we consider the \emph{Arbitrary Pattern Formation} ($\apf$) task by means of a swarm of very weak - in terms of capabilities - robots moving on graphs. Initially, each robot occupies a different vertex of the graph. This task calls for a distributed algorithm that allows a set of autonomous mobile robots \emph{to form any specific but arbitrary geometric pattern given as input}. The pattern formation task is one of the basic primitives extensively studied in the context of robot-based computing systems.
Whether or not a mobile robot system can solve a given problem typically depends on the capabilities one assumes for robots. A common approach in distributed computing is to detect the minimal capabilities that are necessary so as robots can perform basic tasks. The rationale behind this approach is twofold: it is theoretically interesting to answer the minimality question; the weaker the model assumed to solve a task, the wider its applicability, including more powerful robots prone to faults.

\subsection{Robots' model}\label{ssec:model}
In this paper, robots are considered to be:
\begin{itemize}
\item \emph{Anonymous}: no unique identifiers;
\item \emph{Autonomous}: no centralized control;
\item \emph{Dimensionless}: no occupancy constraints, no volume, modeled as entities located on vertices of a graph;
\item \emph{Oblivious}: no memory of past events;
\item \emph{Homogeneous}: they all execute the same \emph{deterministic}\footnote{No randomization features are allowed.} algorithm;
\item \emph{Silent}: no means of direct communication;
\item \emph{Disoriented}: no common coordinate system, no common left-right orientation;
\end{itemize}
%

%
Each robot in the system has sensory capabilities allowing it to determine the location of other robots in the graph, relative to its own location.
Each robot refers in fact to a \emph{Local Coordinate System} (LCS) that might be different from robot to robot. Each robot follows an identical algorithm that is preprogrammed into the robot. 
The behavior of each robot can be described according to the sequence of four states: \Wait, \Look, \Compute, and \Move. Such states form a computational cycle (or briefly a cycle) of a robot.
\begin{enumerate}
\item \Wait. The robot is idle. A robot cannot stay indefinitely idle. 
\item  \Look. The robot observes the environment by activating its sensors which will return a snapshot of the positions of all other robots with respect to its LCS. Each robot is viewed as a point. Hence, the result of the snapshot (i.e., of the observation) is just a set of coordinates in its LCS.
\item  \Compute. The robot performs a local computation according to a deterministic algorithm $\A$ (we also say that the robot executes $\A$). The algorithm is the same for all robots, and the result of the \Compute phase is a destination point along with a path to reach it.
\item  \Move. If the destination point is the current vertex where $r$ resides, $r$ performs a $\nil$ movement (i.e., it does not move); otherwise it moves to the adjacent vertex selected along the computed path. 
\end{enumerate}

When a robot is in \Wait we say it is \emph{inactive}, otherwise it is \emph{active}. In the literature, the computational cycle is simply referred to as the \Look-\Compute-\Move (LCM) cycle, as during the \Wait phase a robot is inactive. 
Initially robots are inactive, but once the execution of an algorithm $\A$ starts - unless differently specified - there is no instruction to stop it, i.e., to prevent robots to enter their LCM cycles. Then, the \emph{termination} property for $\A$ can be stated as follows: once robots have reached the required goal by means of $\A$, from there on robots can perform only the $\nil$ movement. 

During the \Look phase, robots can perceive \emph{multiplicities}, that is whether a same point is occupied by more than one robot. The multiplicity detection capability might be \emph{local} or \emph{global}, depending whether the multiplicity is detected only by robots composing the multiplicity or by any robot performing the \Look phase, respectively. Moreover, the multiplicity detection can be \emph{weak} or \emph{strong}, depending whether a robot can detect only the presence of a multiplicity or if it perceives the exact number of robots composing the multiplicity, respectively. In this work we assume that each robot is endowed with the global strong multiplicity detection.

Concerning the movements, in the graph environment moves are always considered as instantaneous. This results in always perceiving robots on vertices and never on edges during Look phases. Hence, robots cannot be seen while moving, but only at the moment they may start moving or when they arrived. The rationale behind this assumption is that the graph may model a communication network, whereas robots model software agents.

\begin{figure}[t]
\begin{center}
\scalebox{0.60}{\input{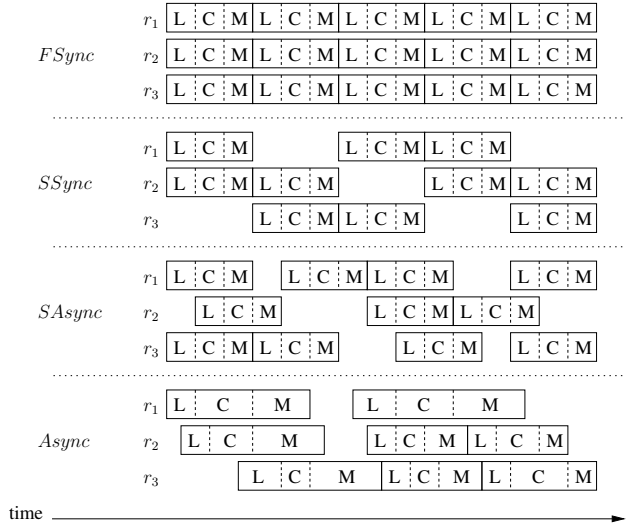}}
\caption{The execution model of computational cycles for each of \fsync, \ssync, \sasync, and \async\ robots. The inactivity of robots is implicitly represented by empty time periods.}\label{fig:models-a}
\end{center}
\end{figure}

We assume that cycles are performed according to the weakest Asynchronous scheduler (\async) (cf.~\cite{BCM16,CDN18d,CDN19,CDN18c,DFSY15,FPSW08,GM10}): the robots are activated independently, and the duration of each phase is finite but unpredictable (the activation of each robot can be thought of as decided by the adversary). As a result, robots do not have a common notion of time. Moreover, according to the definition of the \Look phase, a robot does not perceive whether other robots are moving or not. Hence, robots may move based on outdated perceptions. In fact, due to asynchrony, by the time a robot takes a snapshot of the configuration, this might have drastically changed once the robot starts moving. The scheduler determining the cycles timing is assumed to be fair, that is, each robot becomes active and performs its cycle within finite time and infinitely often. Figure~\ref{fig:models-a} compares the \async scheduler with the other scheduler proposed in the literature. In the figure, the \Wait state is implicitly represented by the time while a robot is inactive. In particular, it shows that in the Fully-synchronous (\fsync) scheduler all robots are always active, and the activation phase can be logically divided into global rounds: for all $i\geq 1$, all robots start the $i$-th LCM cycle simultaneously and synchronously execute each phase. 

The Semi-synchronous (\ssync, cf.~\cite{SY99,YS10,YUKY15}) scheduler coincides with the \fsync model, with the only difference that some robots may not start the $i$-th LCM cycle for some $i$ (some of the robots might be in the \Wait state), but all of those who have started the $i$-th cycle synchronously execute each phase.

The Semi-asynchronous (\sasync, cf.~\cite{CDN18a}) still maintains a sort of synchronous behavior as each phase lasts the same amount of time, but robots can start their LCM cycles at different times. It follows that while a robot is performing a \Look phase, other active robots might be performing the \Compute or the \Move phases. 

Clearly, the four synchronization schedulers induce the following hierarchy (see, e.g.~\cite{CDN18a,DFPSY16,DDFN18}): \fsync robots are more powerful (i.e. they can solve more tasks) than \ssync robots, that in turn are more powerful than \sasync  robots, that in turn are more powerful than \async robots. This simply follows by observing that  the adversary can control more parameters in \async than in \sasync, and it controls more parameters in \sasync than in \ssync and \fsync. In other words, protocols designed for \async robots also work for \sasync, \ssync and \fsync robots. Contrary, any impossibility result stated for \fsync robots also holds for \ssync, \sasync and \async robots.  

In the \async scheduler, the activations of the robots determine specific ordered time instants. Let $C(t)$ be the configuration observed by some robots at time $t$ during their \Look phase, and let $\{t_i : i = 0,1,\ldots \}$, with $t_i < t_{i+1}$,  be the set of all time instances at which at least one robot takes the snapshot $C(t_i)$. Since the information relevant for the computing phase of each robot is the order in which the different snapshots occur and not the exact time in which each snapshots is taken, then without loss of generality we can assume $t_i = i$ for all $i = 0,1,\ldots$. Then, an \emph{execution} of an algorithm $\A$ from an initial configuration $C$ is a sequence of configurations $\Ex : C(0),C(1),\ldots$, where $C(0)=C$ and $C(t+1)$ is obtained from $C(t)$ by moving some robot according to the result of the \Compute phase as implemented by $\A$. Notice that this definition of execution works also for the other schedulers. Moreover, given an algorithm $\A$, in \async (but also in \sasync and \sasync) there exists more than one execution from $C(0)$ depending on the activation of the robots (which depends on the adversary).

\subsection{Previous work}
For robots moving on the Euclidean plane, a restricted version of $\apf$ has been first solved in~\cite{DPV10}. In fact, the proposed algorithm requires at least $n\geq 4$ asynchronous robots endowed with chirality, that is robots share a common handedness. Moreover, the possible patterns exclude the possibility to form multiplicities. The answer to this restricted setting for $\apf$ provided a nice characterization of the problem that was shown to be equivalent to Leader Election within the same set of assumptions. In particular, the configurations from which the proposed algorithm could output any pattern (without multiplicities) are the so-called \emph{leader configurations}. These are configurations of robots (including some symmetric ones) from which it is possible to elect a leader. Attempts to remove those restrictions can be found in~\cite{BT18,YY14}, but randomization techniques are used. In~\cite{CDN19}, instead $\apf$ has been solved by means of a deterministic algorithm, without chirality and allowing multiplicities.
Further investigations of $\apf$ in the Euclidean plane referring to slightly different models can be found in~\cite{BKS20,FPSW05}. It is worth mentioning that when multiplicities are allowed for the patterns, the degenerate case of point formation (aka \emph{Gathering}) is included in $\apf$. Actually, the gathering task has been fully characterized in~\cite{CFPS12}. It constitutes a very special case that deserves main attention.

For robots moving on graphs, and in particular on an infinite square grid, $\apf$ has been recently addressed in~\cite{BAKS19}. The initial configuration is assumed to be asymmetric and still the allowed patterns do not contain multiplicities. Hence, the considered $\apf$, so far does not include gathering. Gathering on infinite or finite square grids has been fully characterized in~\cite{DDKN12,DN17}, also considering the minimization of the overall travelled distances.

\subsection{Our results} 
Our investigation for $\apf$ on graphs has started by considering square grids allowing also multiplicities in the patterns. Then we realized that a natural extension of the problem is to consider any regular tessellation graph as discretization of the Euclidean plane, that is also hexagonal and triangular grids deserve investigation. In particular, the latter can be considered as the most general topology in terms of possible symmetries and trajectories. 
In this paper, we address the resolution of $\apf$, including multiplicities, on all the three regular tessellations by providing a unique algorithm. The algorithm is first described in details with respect to the triangular grid,  when the initial configuration is asymmetric. We follow a formal design and analysis to provide our algorithm, along with the correctness proof. To this aim, we used the design methodology proposed in~\cite{CDN20b}.
Furthermore we revisit the algorithm with respect to both the square and the hexagonal grids, pointing out any possible deviations required with respect to the specific topology.

\subsection{Outline} 
This paper is organized as follows. Next section first formally defines the addressed problem and then it introduces the notation used by the provided algorithm called $\Aform$. Section~\ref{sec:algorithm} provides a high-level description of $\Aform$ designed by also remarking the strategy underlying the algorithm. Section~\ref{sec:formalization} formalizes the algorithm and provides the correctness. Since all the details are given with respect to the triangular grid, in Section~\ref{sec:extensions} we revisit the algorithm with respect to both the square and the hexagonal grids. Section~\ref{sec:conclusion} concludes the paper by highlighting some final remarks. 


\section{Problem definition and basic notation}\label{sec:problem}
The topology where robots are placed on is represented by a simple, undirected, and connected graph $G=(V,E)$, with vertex set $V$ and edge set $E$. 
A function $\lambda: V\to \mathbb{N}$ represents the number of robots on each vertex of $G$, and we call $C=(G,\lambda)$ a \emph{configuration} whenever $\sum_{v\in V} \lambda(v)$ is bounded and greater than zero. A vertex $v\in V$ such that $\lambda(v)> 0$ is said \emph{occupied}, \emph{unoccupied} otherwise. A \emph{multiplicity} occurs in any vertex $v\in V$ such that $\lambda(v)> 1$.

\subsection{Configurations on tessellation graphs}%
In this work we consider $G$ as an infinite graph generated by a \emph{plane tessellation}. 
A tessellation is a tiling of a plane with polygons without overlapping. A \emph{regular tessellation} is a tessellation which is formed by just one kind of regular polygons of side length $1$ and in which the corners of polygons are identically arranged. According to~\cite{GS87}, there are only three regular tessellations, and they are generated by squares, equilateral triangles or regular hexagons (see Fig.~\ref{fig:tessellation}).
An infinite lattice of a regular tessellation is a lattice formed by taking the vertices of the regular polygons in the tessellation as the points of the lattice. A graph $G$ is induced by the point set $S$ if the vertices of $G$ are the points in $S$ and its edges connect vertices that are distance $1$ apart. A \emph{tessellation graph} of a regular tessellation is the infinite graph embedded into the Euclidean plane induced by the infinite lattice formed by that tessellation~\cite{Ionascu12}. We denote by $\GS$ ($\GT$ and $\GH$, resp.) the tessellation graphs induced by the regular tessellations generated by squares (equilateral triangles and regular hexagons, resp.). In this work we consider configurations $C=(G,\lambda)$ where $G\in \{\GS,\GT,\GH\}$. 


\begin{definition}\label{def::canonical}
Given a graph $G\in \{\GS,\GT,\GH\}$, any line parallel to a subset of edges of $G$ is called a \emph{canonical direction}. The smallest angle formed by the available canonical directions is called the \emph{canonical angle}.
\end{definition}

According to Definition~\ref{def::canonical}, in $\GS$ there are just two canonical directions and the canonical angle is of $90^\circ$. In both $\GT$ and $\GH$ there are three canonical directions and the canonical angle is of $60^\circ$.

\begin{figure}[t]
\centering
\includegraphics[scale=0.75]{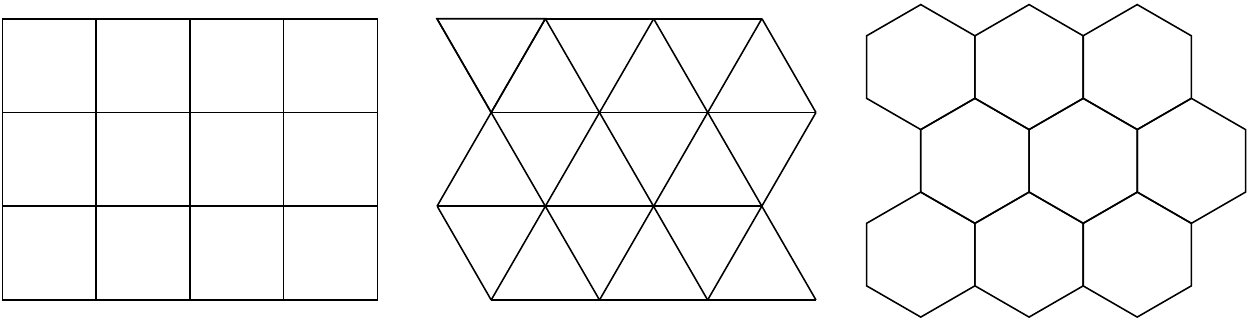}
\caption{Part of regular plane tessellations.}
\label{fig:tessellation}
\end{figure}

\subsection{Configuration automorphisms and symmetries}
Two undirected graphs $G=(V,E)$ and $G'=(V',E')$ are \emph{isomorphic} if there is a bijection $\varphi$ from $V$ to $V'$ such that $\{u,v\} \in E$ if and only if $\{\varphi(u),\varphi(v)\} \in E'$. An \emph{automorphism} on a graph $G$ is an isomorphism from $G$ to itself, that is a permutation of the vertices of $G$ that maps edges to edges and non-edges to non-edges. The set of all automorphisms of $G$, under the composition operation, forms a group called \emph{automorphism group} of $G$ and denoted by $\Aut{G}$. If $|\Aut{G}|=1$, that is $G$ admits only the identity automorphism, then $G$ is said \emph{asymmetric}, otherwise it is said \emph{symmetric}. Two distinct vertices $u,v\in V$ are \emph{equivalent} if there exists an automorphism $\varphi\in \Aut{G}$ such that $\varphi(u)=v$.

The concept of graph automorphism can be extended to configurations in a natural way: (1) two configurations $C=(G,\lambda)$ and $C'=(G',\lambda')$ are isomorphic if $G$ and $G'$ are isomorphic via an automorphism $\varphi\in \Aut{G}$ and $\lambda(v)=\lambda'(\varphi(v))$ for each vertex $v$ in $G$; (2) an automorphism of a configuration $C=(G,\lambda)$ is an isomorphism from $C$ to itself, and (3) the set of all automorphisms of $C$ forms a group under the composition operation that we call automorphism group of $C$ and denote as $\Aut{C}$. 
Moreover, if $|\Aut{C}|=1$ we say that $C$ is \emph{asymmetric}, otherwise it is \emph{symmetric}. Two distinct robots $r$ and $r'$ in a configuration $(G,\lambda)$ are \emph{equivalent} if there exists $\varphi\in \Aut{C}$ that makes equivalent the vertices in which they reside. Note that $\lambda(u)=\lambda(v)$ whenever $u$ and $v$ are equivalent. Moreover, if $u$ and $v$ are equivalent, a robot $r$ cannot distinguish its position at vertex $u$ from robot $r'$ located at vertex $v = \varphi(u)$. As a consequence, no algorithm can distinguish between two equivalent robots. 

%
In general, no algorithm can avoid that the two equivalent \async robots start the computational cycle simultaneously. In such a case, there might be a so called \emph{pending move} or \emph{pending robot}, that is one of the two robots performs its entire computational cycle while the other has not started or not yet finished its Move phase. Formally, a robot $r$ is pending in a configuration $C(t)$, if at time $t$ robot $r$ is active, has taken a snapshot $C(t')\neq C(t)$ with $t' < t$, and is planning to move or is moving with a non-nil trajectory. 
Clearly, any other robot $r'$ is not aware whether there is a pending robot $r$, that is it cannot deduce such an information from the snapshot acquired in the Look phase. This fact greatly increases the difficulty to devise algorithms for symmetric configurations. Notice that all such difficulties are completely removed if an algorithm produces always \emph{stationary} configurations: a configuration $C(t)$ is called \emph{stationary} if there are no pending robots in $C(t)$. A way to produce stationary configurations is to guarantee that an algorithm always moves one robot at a time.

\smallskip
Concerning the configurations addressed in this work, it is not difficult to see that any $C=(G,\lambda)$, with $G\in \{\GS,\GT,\GH\}$, admits two types of automorphisms only: \emph{reflections}, defined by a reflection axis which acts as a mirror; \emph{rotations}, defined by a center and an angle of rotation. All the reflection axes are of two types: the reflection axes of the considered regular polygons and those coincident with any side of the regular polygons. The centers of possible rotations can be located only on specific points of the regular polygons: on the center, on one vertex, or on the middle point of a side. The rotation angle is specific of each given tessellation graph. 

\subsection{The Arbitrary Pattern Formation ($\apf$) problem}%
A configuration $C=(G,\lambda)$, with $G=(V,E)$, is \emph{initial} if both the following conditions hold: (1) each robot is idle and placed on a different vertex, that is $\lambda(v)\leq 1$ for each $v\in V$; (2) $C$ is asymmetric. The set containing all the initial configurations is denoted by $\I$.

The goal of the $\apf$ problem is to design a distributed algorithm $\A$ that guides the robots to form a fixed arbitrary pattern $F$ starting from any configuration $C=(G,\lambda)$ such that $G\in \{\GS,\GT,\GH\}$ and $C\in \I$. The pattern $F$ is a multiset of vertices, given in any coordinate system, indicating the corresponding target vertices in the tessellation graph $G$. It constitutes the input for all robots. Due to absence of a common global coordinate system, the robots decide that the pattern is formed when the current configuration becomes ``similar'' to $F$ with respect to translations, rotations, reflections. 
The problem can be formalized as follows: 
an algorithm $\A$ solves the $\apf$ problem for an initial configuration $C$ if, for each possible execution $\Ex : C=C(0),C(1),\ldots$ of $\A$, there exists a finite time instant $t^*>0$ such that $C(t^*)$ is similar to $F$ and no robot moves after $t^*$, i.e., $C(t) = C(t^*)$ holds for all $t\ge t^*$. 

\subsection{Notation}
Here we introduce some concepts and notation used to describe the proposed algorithm. 
Given a configuration $C=(G,\lambda)$, we use $R=\{r_1,r_2,\ldots,r_n\}$ to denote the set containing all the $n$ robots located on $G$ (we recall that robots are anonymous and such a notation is used only for the sake of presentation). The distance $d(u,v)$ between two vertices $u,v\in V$ is the number of edges of a shortest path connecting $u$ to $v$. We extend the notion of distance to robots: $d(r_i,r_j)$ denotes the distance between the two vertices in which the robots reside. Symbol $D(r)$ is used to denote the \emph{sum of distances} of $r\in R$ from any other robot, that is $D(r) = \sum_{r_i\in R\setminus \{r\}} d(r,r_i)$.  
 
Given a set of points $P$ in the plane, $\mbr(P)$ represents the \emph{minimum bounding rectangle} of $P$, that is the rectangle enclosing all the points in $P$ defined as follows: its sides are parallel to the Cartesian axes and each pair of parallel sides are as close as possible. According to the definition we get that $\mbr(P)$ is unique. 
This definition can be easily extended to a set of robots $R$ placed on the tessellation graph $\GS$ where the canonical directions are just two, and they can naturally play the role of the Cartesian axes. Unfortunately, it does not work when $R$ is placed on tessellation graphs such as $\GT$ or $\GH$. To generalize it, we move to the concept of \emph{bounding parallelogram} $\bp(R)$, defined as any parallelogram enclosing all robots, with sides parallel to two of the three available canonical directions, and with each pair of parallel sides as close as possible. Since $\GT$ or $\GH$ admit three canonical directions, it can be observed that the bounding parallelogram of $R$ is not unique. In fact, there are three possible bounding rectangles (e.g., see Fig.~\ref{fig:mbr}).  

%
Given any $\bp(R)$, we denote by $\h(\bp(R))$ and $\w(\bp(R))$, with $\h(\bp(R)) \leq \w(\bp(R))$, the width and height of $\bp(R)$, respectively. 
Similarly, $\h(\bp(F))$ and $\w(\bp(F))$ 
are used to denote the same values with respect to $\bp(F)$. 

\begin{figure}[t]
   \graphicspath{{fig/}}
   \centering
   \def\svgwidth{\columnwidth}
   {\large\scalebox{0.6}{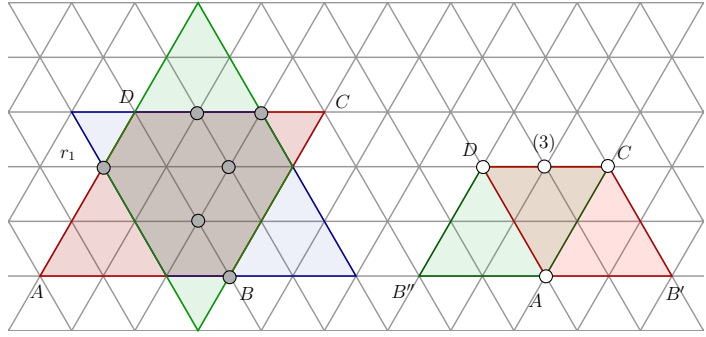}}
   \caption{%
\textit{(left)} An initial configuration $C$ with $n=6$ robots. It shows that $\bp(R)$ is not unique in $\GT$.  The red parallelogram generates the \LSS. The leading corner is $A$ and the leading direction is $AB$. The unique \LSS is $\ell = (0,0,0,1,~0,0,1,0, ~1,0,1,0, ~0,1,1,0)$. 
Notice that robot $r_1$ has maximum sum of distances, with value $D(r_1)=13$. \textit{(right)} A possible pattern $F$ to be formed. The number close to a vertex refers to a multiplicity. Since $F$ is symmetric, there are two (equivalent) $\mbp(F)$.  
}
\label{fig:mbr}
\end{figure}


Let $\bp(R)$ be any bounding parallelogram of $R$. We associate a sequence of integers to each canonical corner of $\bp(R)$ (e.g., corners $A$ and $C$ in Fig.~\ref{fig:mbr}). The sequence associated with a canonical corner $A$ is defined as follows. Scan the finite grid enclosed by $\bp(R)$ from $A$ along $h(\bp(R))$ (say, from $A$ to $B$) and sequentially all grid lines parallel to $AB$ in the same direction. For each grid vertex $v$, put $\lambda(v)$ in the sequence. Denote the obtained sequence as $s(AB)$. 
Being $h(\bp(R)) = w(\bp(R))$ in the example, from $A$ it is also possible to obtain the sequence $s(AD)$, and hence four sequences can be defined in total, two for the corner $A$ and two for the corner $C$. If any two of these sequences are equal, then it implies that the configuration admits a (reflectional or rotational) symmetry. We denote by \LSS the lexicographically smallest sequence. It is unique by definition.
 
The canonical corner from which an \LSS starts is called the \emph{leading corner}; the canonical direction from the leading corner used to create the \LSS is called the \emph{leading direction}. The \LSS of a given $\bp(R)$ is denoted as $\ell(\bp(R))$, or simply as $\ell$ when $\bp(R)$ can be inferred by the context. 

\begin{definition}\label{def:mbr}
Let $C=(G,\lambda)$ be a configuration with $G\in \{\GS,\GT,\GH\}$ and set of robot $R$. A \emph{minimum bounding parallelogram} $\mbp(R)$ is defined as any parallelogram $\bp(R)$ with sides parallel to two canonical directions of $G$, with $h(\bp(R))$ minimum, and with minimum \LSS in case of ties.
\end{definition}
Any asymmetric configurations admits exactly one $\mbp(R)$ whereas symmetric configurations admit multiple $\mbp(R)$'s. However, the \LSS's associated to such $\mbp(R)$'s are all the same.
%

\section{Description of the algorithm}\label{sec:algorithm}
In this section, we provide a high-level description of our algorithm $\Aform$ designed to solve $\apf$ for any initial configuration $C=(\GT,\lambda)$ composed of $n$ \async robots endowed with the global strong multiplicity detection and with all the minimal capabilities recalled in Section~\ref{ssec:model}. We assume $n\ge 3$, since for $n = 1$ the $\apf$ problem is trivial and for $n = 2$ we get that $C$ is symmetric. Concerning the pattern $F$, it might contain multiplicities. 

\subsection{The strategy}\label{ssec:strategy}
In general, a single robot has rather weak capabilities with respect to the general problem it is asked to solve along with other robots (we recall that robots have no direct means of communication). For this reason, any resolution algorithm should be based on a preliminary decomposition approach: the problem should be divided into a set of sub-problems so that each sub-problem is simple enough to be thought of as a ``task'' to be performed by (a subset of) robots. This subdivision could require several steps before obtaining the definition of such simple tasks, thus generating a sort of hierarchical structure.

Following this approach, $\apf$ is initially divided into four sub-problems denoted as Reference System ($\RS$), Partial Pattern Formation ($\PPF$), Finalization ($\Fin$), and Termination ($\Term$). Some of these sub-problems are further refined until the corresponding tasks can be suitably formalized according to the assumed capabilities of the robots. This leads to the following decomposition:
\begin{itemize}
\item
\emph{Reference System ($\RS$ = How to embed $F$ on $\GT$)}. This sub-problem concerns one of the main difficulties arising when the general pattern formation problem is addressed: the lack of a unique embedding of $F$ on $G_T$ that allows each robot to uniquely identify its target (the final destination vertex to form the pattern). In particular, $\RS$ can be described as the problem of moving or matching some (minimal number of) robots into specific positions such that they can be used by any other robot as a common reference system. These robots are called \emph{guards}. The realized reference system should imply a unique mapping from robots to targets, and this mapping should be maintained along all the movements of robots. In our strategy $\RS$ is further divided into three sub-problems denoted as $\RS_{1a}$, $\RS_{1b}$, and $\RS_{2}$. These sub-problems are simple enough to be associated to three tasks named $T_1$, $T_2$ and $T_3$, respectively. The first two are devoted to place the first guard denoted as $r_1$, whereas the third fixes the position of a second guard denoted as $r_n$. Once such positions are reached by the two guards, the requested reference system is given by two lines passing through the vertices occupied by the guards and forming a canonical angle between them. When the reference system is created, all the robots except the guards result to be located in a specific quadrant called $Q^-$.
\item
\emph{Partial Pattern Formation ($\PPF$ = How to form part of $F$)}. 
This sub-problem is associated with task $T_4$ and it is addressed only once $\RS$ is solved.  It concerns the formation of a pattern similar to part of $F$ by using robots in $R''=R\setminus \{r_1,r_n\}$ only.  Thanks to the common reference system, all robots can agree on embedding $F$ on a quadrant denoted as $Q^+$ and different from $Q^-$.  During the task, all the $n-2$ robots in $R''$  will be moved from $Q^-$ to the quadrant $Q^+$. Robots are moved one at a time so that no undesired collisions are created.
\item
\emph{Finalization ($\Fin$ = How to finally move $r_1$ and $r_n$ so that $F$ is formed)}. 
It refers to the so-called finalization task and occurs when the only robots not well positioned according to $F$ are the guards. It is worth to mention that while moving guards $r_1$ and $r_n$, the common reference system is lost. However, we are able to guarantee that robots can always detect they are solving $\Fin$ and that the two robots, by performing ad-hoc movements, can reach their targets so that the pattern $F$ is correctly completed. $\Fin$ is divided into three tasks: $T_5$ concerns the movement of $r_n$, whereas $T_6$ and $T_7$ are related to the movement of $r_1$. 
\item
\emph{Termination ($\Term$)}. It refers to the requirement of letting robots recognize the pattern has been formed, hence no more movements are required. In our strategy, a task $T_8$ is designed to address this problem. Clearly, only $\nil$ movements are allowed and it is not possible to switch to any other task.
\end{itemize}
In the remainder of the section we provide details for each designed task.

\subsection{Task $T_1$}\label{ssec:T1}
It selects a robot denoted as $r_1$ (the first guard) such that $D(r_1)$ is maximum (cf Fig.~\ref{fig:mbr}). In case of ties, $r_1$ has the minimum position in $\ell(\mbp(R))$ -- recall that the input configuration is asymmetric and hence $\mbp(R)$ is unique. Let $R'=R\setminus\{r_1\}$, during this task $r_1$ moves through any shortest path toward to the closest vertex that satisfies the following Boolean variable:

\begin{itemize}
\item 
     \emph{$\vguno$ = exists a unique line parallel to a canonical direction passing through $r_1$ and each $\bp(R')$.}
\end{itemize}
Note that, when $\vguno$ holds we identify the unique line passing through $r_1$ and each $\bp(R')$ as the \emph{line induced} by $\vguno$.

\begin{figure}[t]
   \graphicspath{{fig/}}
   \centering
   \def\svgwidth{\columnwidth}
   \large\scalebox{0.8}{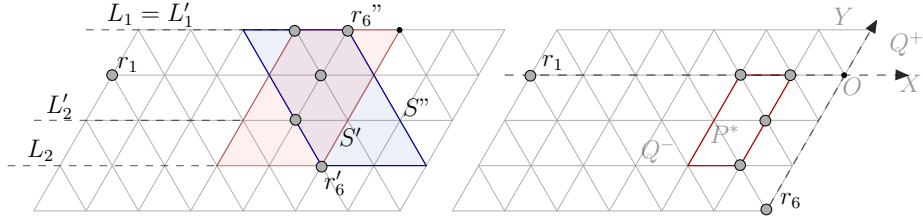}
 \caption{%
\textit{(left)} Visualization of sub-problem $\RS_{1a}$ concerning the initial movement of $r_1$ (cf. configuration $C$ in Fig.~\ref{fig:mbr}). 
\textit{(right)} Visualization of sub-problem $\RS_{1b}$ concerning the final destination of $r_1$. Once $r_1$ stops, all the items necessary to define the reference system can be settled (cf Remark~\ref{rem:afterT2}).
}
\label{fig:T1}
\end{figure}

%
%

\subsection{Task $T_2$}\label{ssec:T2}
In this task we assume true the variable $\vguno$ holding at the end of task $T_1$ - this must be intended as a \emph{pre-condition} imposed by our strategy in order to correctly perform $T_2$. The main aim of this task is twofold: (1) to move $r_1$ so that its position allows to define the $X$-axis and (2) to identify the second guard $r_n$. Anyway, additional properties are guaranteed when the task is completed. 

When the task starts, the role of $r_1$ is assigned to the robot $r$ such that $D(r)$ is maximum whereas the second guard $r_n$ is identified as follows:
\begin{itemize}
\item
Let $L$ be the line induced by $\vguno$.
It can be observed that there are exactly two distinct $\bp(R')$'s with sides parallel to $L$. Let $L_1$ and $L_2$ be the two lines parallel to $L$ shared by the two $\bp(R')$'s (cf Fig.~\ref{fig:T1}). Denote the two $\bp(R')$'s as $P'$ and $P''$, and denote as $S'$ ($S''$, resp.) the side of $P'$ ($P''$, resp.) which lies neither on $L_1$ nor on $L_2$ and is further from $r_1$. 
In particular, $P'$ ($P''$, resp.) is the parallelogram having the canonical angle formed by the intersection of $S'$ ($S''$, resp.) and $L_1$ ($L_2$, resp.) - the red parallelogram in Fig.~\ref{fig:T1}. 
 Denote as $r'_n$ ($r''_n$, resp.) the robot on $S'$ ($S''$, resp.) closest to $L_1$ ($L_2$, resp.). The second guard useful to define the reference system is selected between $r'_n$ and $r''_n$.

Robot $r_1$ considers the line $L'_1$ ( $L'_2$, resp.) defined as $L_1$ ( $L_2$, resp.) but referred to $R'\setminus \{r'_n\}$ ( $R'\setminus \{r''_n\}$, resp. ) instead of $R'$. Then $r_1$ selects the closest line between $L'_1$ and $L'_2$ (it arbitrarily selects one of the two in case of ties). Without loss of generality, assume that $r_1$ selects $L'_1$. According to this choice, $r_1$ promotes $r'_n$ to be $r_n$, that is the second guard (symmetrically, if $r_1$ selects $L_2$, then $r''_n$ is promoted). 
\end{itemize}
After computing the second guard, the robots have enough information to identify the requested common reference system, as remarked in the following statement.

\begin{remark}\label{rem:afterT2}
\emph{
After computing the second guard, robots have sufficient information to compute a common reference system. In fact, the line 
between $L'_1$ and $L'_2$ selected by $r_1$ and letting all robots in $R'$ in the same half-plane defines the $X$-axis, and this axis must be intended as directed from $r_1$ to all the other robots; all vertices in the half-plane containing robots in $R'$ are considered with negative $Y$-coordinates. The second guard $r_n$ is induced by the line between $L'_1$ and $L'_2$ selected by $r_1$ (as described above). The line passing through $r_n$, intersecting the $X$-axis, and forming a canonical angle in the first quadrant defines the $Y$-axis. Finally, the intersection between the two axes defines the origin of the system denoted as $O$. In this reference system, the fist quadrant is denoted as $Q^+$, while the third quadrant is denoted as $Q^-$. 
}
\end{remark}
Now, concerning the current task, it remains to be defined the correct positioning of both guards. The target of $r_1$ ($r_n$, resp.) is on the $X$-axis ($Y$-axis, resp.) so that the distance from the origin ensures that the configuration remains asymmetric during the subsequent $\PPF$ task. To define such a distance, robots compute the following:
\begin{itemize}
\item
Let $R^*$ be the (possibly empty) subset of robots of $R''$ lying in $Q^-$, $\PP$ be the parallelogram $\bp(R^*)$ with the constraint that it must use the directions parallel to the $X$- and $Y$-axes, and let $\Delta = \max \{ w(\PP), w( \mbp(F))\}$.\footnote{This definition of $\Delta$ is given with respect to $R''$ instead of $R'$ so that it can be also used in the subsequent tasks $T_3$, $T_4$, and $T_5$.}
\end{itemize}
According to $\Delta$, the target of $r_1$ corresponds to the closest vertex on the $X$-axis which is at distance at least $3\Delta$ from the origin. The trajectory followed is represented by any shortest path to the target. 
Note that at the and of task $T_2$, variable $\vguno$ still holds, but the movement of $r_1$ makes true the following additional variables:
\begin{itemize}
\item 
     \emph{$\vhpdue$ = all the robots in $R''$ are in the same half-plane 
     with respect to the line induced by $\vguno$.}
\item
     \emph{$\vdruno = d(r_1,O)\ge 3\Delta$.}
\end{itemize}

\begin{figure}[t]
   \graphicspath{{fig/}}
   \centering
   \def\svgwidth{\columnwidth}
   {\large\scalebox{0.7}{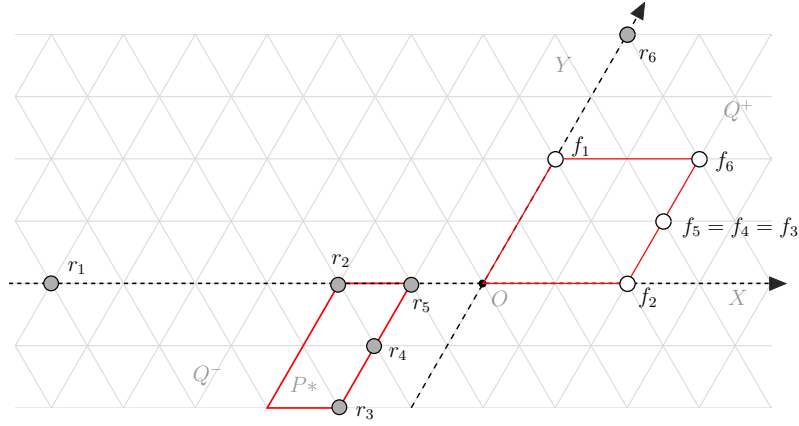}}
  \caption{%
Visualization of sub-problem $\RS_{2}$ concerning the placing of guard $r_n = r_6$. Notice the embedding $F_e$ in $Q^+$  (cf Definition~\ref{def:embedding}) of the pattern represented in Fig.~\ref{fig:mbr} and the ordering of all robots in $R''$ according to the lexicographic order of the coordinates of the vertices in which they reside.
}
\label{fig:T2.1}
\end{figure}

%

\subsection{Task $T_3$}\label{ssec:T3}
The aim of this task is to locate $r_n$ to a destination easily recognizable in the next tasks, especially during the formation of the (sub-) pattern by robots in $R''$. As a pre-condition, in this task we assume true all the variables holding at the end of task $T_2$, namely $\vguno$, $\vhpdue$, and $\vdruno $.

According to the pre-condition, in this task robots can use Remark~\ref{rem:afterT2}, with the  difference that now the $X$-axis is directly defined as the direction induced by $\vguno$. By using that remark, robots can identify both guards and re-compute the common reference system. At this point, $r_n$ performs the task by simply moving along the $Y$-axis (cf Fig.~\ref{fig:T2.1}) toward the closest vertex $(0,y)$ such that the following variable holds:
\begin{itemize}
\item
     \emph{ $\vgenne$ = $r_n$ is at a vertex $(0,y)$, with $2\Delta \leq y < d(r_1,O)$. }
\end{itemize}
The following additional remark states how robots can re-compute the common reference system in the subsequent tasks.
\begin{remark}\label{remark:RS2}
\textit{
At the end of Task $T_3$, i.e. when both the guards are suitably placed, each robot can recognize the formed reference system: the two guards can be detected according to function $D()$, since $r_1$ and $r_n$ have the largest and second largest value of $D()$, respectively; if $\vguno$ holds, the  induced line defines the $X$-axis directed from $r_1$ to all the other robots; the $Y$-axis is the line passing through $r_n$, intersecting the $X$-axis, directed from the intersection toward $r_n$, and forming a canonical angle in the first quadrant. Finally, the fact that the two guards are correctly positioned according to the strategy can be verified according to $\Delta$, since $Q^-$ (which contains all robots in $R''$) is identified. 
}
\end{remark}


\subsection{Task $T_4$}\label{ssec:T4}
This task concerns the so called ``Partial Pattern Formation'', that is forming part of the input pattern $F$ by using robots in $R''=R\setminus \{r_1,r_n\}$ only. To this aim, all the $n-2$ robots in $R''$ initially located in the quadrant $Q^-$ will be moved in the quadrant $Q^+$. In our strategy, it is addressed only once $\RS$ is solved, that is when the two guards $r_1$ and $r_n$ are suitably placed. More formally, as a pre-condition for performing this task our algorithm requires that $\vguno$, $\vgenne$, and $\vdruno$ are all true. 

It is clear that this task can be accomplished by robots in $R''$ only if they know the common reference system: this can be obtained as described in Remark~\ref{remark:RS2}. 
%
Concerning the partial pattern to be formed, all robots must agree on the positions they have to reach in $Q^+$; this problem is solved by performing an embedding of $F$ into $\GT$ according to the following definition. 
\begin{definition}[Embedding of the pattern]\label{def:embedding}
$F_e$ is the set of vertices in $Q^+$ obtained by translating $F$ so that the following conditions hold: (1) the leading corner of $\mbp(F)$ is mapped onto the origin $O$, and (2) the leading direction of $\mbp(F)$ coincides with the positive direction of the $Y$-axis.
\end{definition}
An example of $F_e$ is shown in Fig.~\ref{fig:T2.1}. Once the robots agree on $F_e$, the main difficulties in this task are to preserve the reference system (induced by guards $r_1$ and $r_n$) and to avoid undesired collisions during the movements. To avoid collisions, robots are moved one at a time according to a schedule induced by the following definitions:
\begin{itemize}
\item
Vertices in $F_e$ are ordered according to the lexicographic order of their coordinates expressed according to the formed $X$- and $Y$-axes. Hence, from now on we denote $F_e$ as the multiset\footnote{Recall that $F$ may contain multiplicities.} 
$\{f_1,f_2,\ldots, f_n\}$, where $i \le j$ if and only if the coordinates of $f_i$ precede those of $f_j$. Similarly for robots in $R''$: they are ordered according to the lexicographic order of the coordinates of the vertices in which they reside and $R''=\{r_2,r_3,\ldots,r_{n-1}\}$.
\item
Vertices $f_1$ and $f_n$ are not used during the resolution of $\PPF$ since they are considered as the final \emph{targets} for the guards. In particular, in the last part of the resolution algorithm, $r_1$ will be moved in $f_1$ and $r_n$ will be moved in $f_n$.
\item
A vertex $f_i\in F_e$, $2\le i\le n-1$, is called the \emph{largest unmatched target} if it is unoccupied whereas $f_{j}$ is occupied for each $i <j < n$.
\item
A robot $r_i\in R''$, $2\le i\le n-1$, is called \emph{largest unmatched robot} if $f_i$ is the largest unmatched target.
\end{itemize}
Algorithm $\Aform$ moves robots in $R''$ in order, moving each time the largest unmatched one toward the largest unmatched target in $F_e$. The trajectory of a moving robot is given by any shortest path leading to its target.

During the task, all the unmatched robots must result to be correctly positioned with this strategy. This is controlled by the following variable:
\begin{itemize}
\item
     \emph{ $\vrpeffe = $ the largest unmatched robot $r_i$ is on a shortest path  
                        from any vertex in $Q^-$ to $f_i$, and each robot $r_j$,  
                        $j<i$, is in $Q^-$. }
\end{itemize}
Notice that at the end of this task, variable $\vguno$ still holds. Finally, also the following additional variables hold:
\begin{itemize}
\item
     \emph{ $\vhpuno = $ all the robots in $R'$ are in the same 
            half-plane with respect to the line induced by $\vguno$. }
\item
     \emph{ $\vpfenne =$ there exists an embedding of $F$ such that all robots 
            in $R''$ are similar to  $F \setminus \{f_1,f_n\}$}
\end{itemize}
\begin{figure}[t]
   \graphicspath{{fig/}}
   \centering
   \def\svgwidth{\columnwidth}
   {\large\scalebox{0.65}{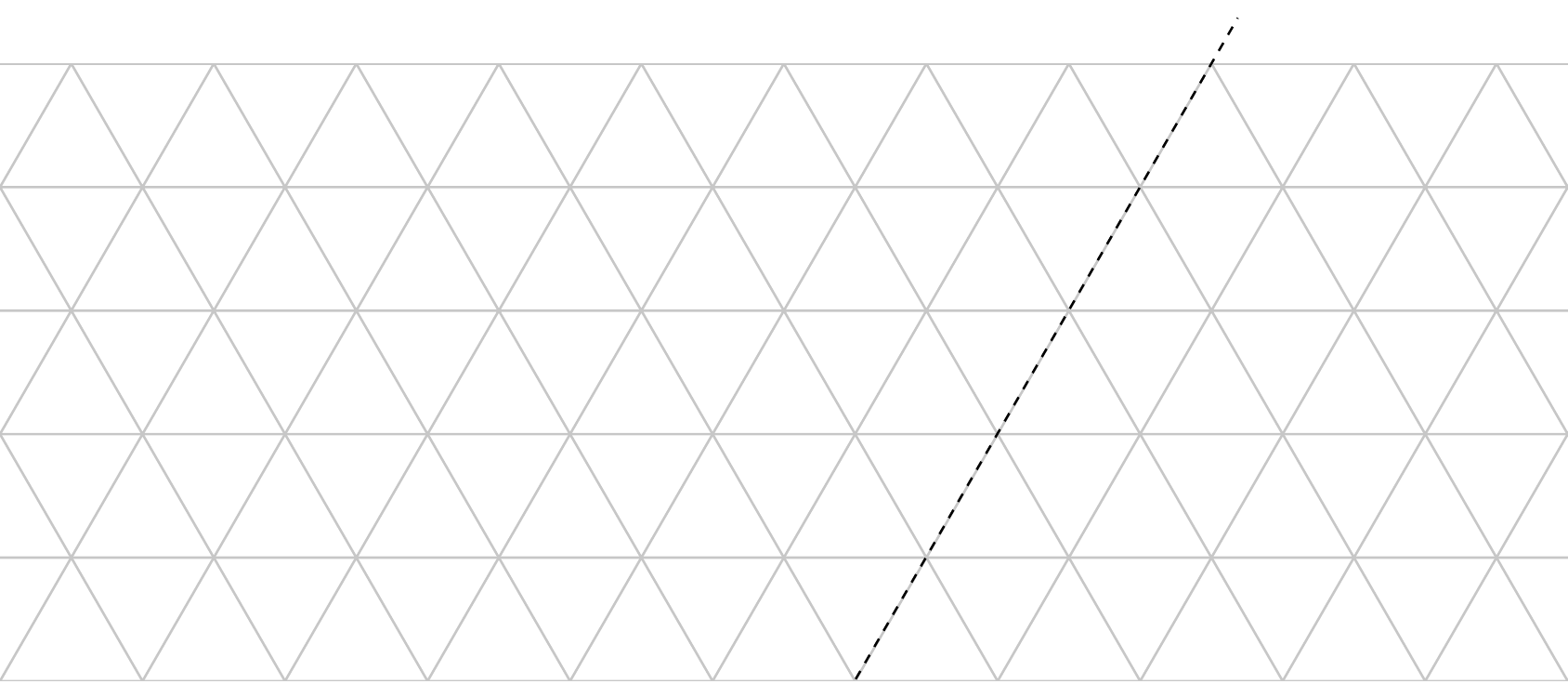}}
  \caption{%
Visualization of the configuration obtained at the end of task $T_4$ (cf Fig.\ref{fig:T2.1}). Gray (black, white, resp.) circles represent unmatched robots (matched robots, unmatched targets, resp.), while integers close to  matched robots refer to multiplicities. 
}
\label{fig:T3.1}
\end{figure}

%
\subsection{Task $T_5$}\label{ssec:T5}
This task is the first associated with ``Finalization'' sub-problem. In particular, it concerns the movement of $r_n$ toward $f_n$. According to the strategy, our algorithm assumes that all the variables made true by the previous task $T_4$ are true, namely $\vguno$, $\vhpuno$, $\vdruno$, and  $\vpfenne$.

Robot $r_n$ moves from the $Y$-axis straightly along the canonical direction parallel to the $X$-axis until a vertex with the same $X$-coordinate of $f_n$ is reached, and then it directly proceeds toward the target (cf. Fig.~\ref{fig:T3.1}). The movement of $r_n$ can take many $\LCM$ cycles and hence it is required a new variable to check the correct positioning of $r_n$:
\begin{itemize}
\item
     \emph{ $\vhrenne = $ point $f_n=(x,y)$ and robot $r_n=(x',y')$, with $x'\leq x$ and $y'\geq y$. }
\end{itemize}
Clearly, the check of variable $\vhrenne$ requires the reference system. However, this cannot be evaluated as it has been done in the previous tasks, since $r_n$ is currently moving (i.e., guards are not suitably placed anymore). Anyway, since $\vpfenne$ holds, the reference system can be deduced from the embedding. In particular:

\begin{remark}\label{remark:T5}
\emph{
During Task $T_5$, each robot can recognize the formed reference system: $r_1$ can be detected according to function $D()$ since it has the largest value of $D()$; if $\vguno$ and $\vhpuno$ hold, the induced line defines the $X$-axis directed from $r_1$ to all the other robots; from $\mbp(F)$ it is possible to check whether its leading corner placed on a vertex $v$ on the $X$-axis makes $f_2$, $\ldots$, $f_{n-1}$  matched, hence defining $r_n$; the $Y$-axis is assumed as the canonical direction passing through $v$ and forming a canonical angle in the first quadrant which contains all robots in $R'$; finally, knowing $r_1$ and $r_n$ and $Q^-$ it is possible to compute $\Delta$ and hence check whether $\vdruno$ holds.}
\end{remark}

Once $r_n$ reaches $f_n$, variables $\vguno$, $\vdruno$ still hold and a new variable is made true:
\begin{itemize}
\item
     \emph{ $\vpfuno = $ pattern $F \setminus \{f_1\}$ formed.}
\end{itemize}
\subsection{Tasks $T_6$ and $T_7$}\label{ssec:T6-T7}
When $T_6$ starts as a consequence of the termination of $T_5$, variables  $\vguno$, $\vdruno$, and $\vpfuno$ are all true. This means that $n-1$ robots reached their targets apart from one robot which is far enough from the others in order to induce just one direction toward the remaining robots. From now on such a robot is referred to as $r_1$.  Since $T_6$ and $T_7$ refer to the ``Finalization'' sub-problem, $r_1$ must be moved toward its target in order to finalize the pattern $F$. This means that during such tasks both guards are no longer correctly positioned and hence the common reference system is no longer available. In particular, the origin $O$ of the system is not defined and hence $\vdruno$ cannot be evaluated. Anyway, we will see that the algorithm will move $r_1$ so that the following variable remains valid during $T_6$:
\begin{itemize}
\item
     \emph{ $\vdrunop =$ the distance between $r_1$ and the other robots guarantees 
                         that $d(r_1,\mbp(F))\ge 3w(\mbp(F))$. }            
\end{itemize}
In particular, task $T_6$ is meant to move $r_1$ toward a target vertex $t$ so that the following properties hold: (1) $\vdrunop$ remains true, and (2) in the subsequent task $T_7$, starting from $t$, robot $r_1$ can reach its final destination $f_1$ by moving straightly along one canonical direction. 

\smallskip
Even if in both $T_6$ and $T_7$ the common reference system is no longer available, robots can take advantage of the existing positions of $r_1$ and of the other robots to correctly finalize the pattern. In particular, when $T_6$ starts the variables $\vguno$, $\vdruno'$, and $\vpfuno$ are all true and, accordingly, robots can compute the following data:

\begin{itemize}
\item
let $U$ be the direction induced by variable $\vguno$. Let $L_1$ and $L_2$ be the lines parallel to $U$, closest to each other, and enclosing $R’=R\setminus \{r_1\}$;
\item
consider the smallest parallelogram $P_1$ ($P_2$, resp.) such that: it encloses the whole set $R$, it guarantees $h(P_1)=h(\mbp(F))$, it has the longest side on $L_1$ ($L_2$, resp.), and it determines a corner $O_1$ ($O_2$, resp.) at the intersection vertex with the shortest side passing through $r_1$ that admits a canonical angle; 
\item
compute $sP_1$ ($sP_2$, resp.) as the sequence of integers associated with $O_1$ ($O_2$ resp.) such that $r_1$ is met as the first robot.
\end{itemize}

\begin{figure}[t]
   \graphicspath{{fig/}}
   \centering
   \def\svgwidth{\columnwidth}
   {\large\scalebox{0.75}{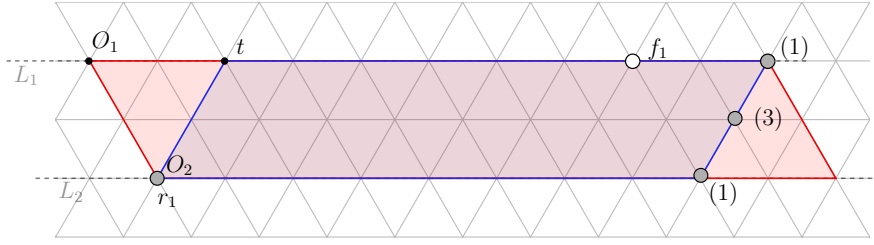}}
  \caption{%
Visualization of the configuration obtained at the end of task $T_5$ (cf Fig.\ref{fig:T3.1}). Gray (black, white, resp.) circles represent unmatched robots (matched robots, unmatched targets, resp.), while integers close to  matched robots refer to multiplicities. 
}
\label{fig:T6}
\end{figure}

%
As an example, $P_1$ and $P_2$ correspond to the red and blue parallelograms represented in Fig.~\ref{fig:T6}, respectively. According to such data, robots can verify whether the current configuration is coherent with task $T_6$ of our algorithm $\Aform$ by performing the following check:
\begin{itemize}
\item
at least one parallelogram between $P_1$ and $P_2$ must be coherent with the $n-1$ elements of $F$ already matched. This means that the last values of $sP_1$ (or $sP_2$) must coincide with the sequence $\ell(\mbp(F))$ except for just one value corresponding to $f_1$ (the vertex to be matched by $r_1$). Notice that it is possible that this happens for both $sP_1$ and $sP_2$ when there is a reflection axis for $F\setminus \{f_1\}$ parallel to the direction $U$.
\end{itemize}
If we denote by $\LSF$ the sequence of integers obtained from $\ell(\mbp(F))$ by decreasing by one the first non-zero element,%
\footnote{Basically $\LSF$ denotes the sequence $\ell(\mbp(F))$ by ignoring $f_1$.}
and by $d_f$ the position in $\ell(\mbp(F))$ of such an element, the above check can be done by better formalizing variable $\vpfuno$:
\begin{itemize}
\item
     \emph{ $\vpfuno =$ there exists $s\in \{sP_1,sP_2\}$ such that $s = s’ + \LSF$, 
                        for some $s’$ made of only $0$'s and just one $1$ in position
                        $d_{r_1}$ and $d_{r_1} < d_f$. }            
\end{itemize}
Referring to the example shown in Fig.~\ref{fig:T6}, variable $\vpfuno$ is made true by the sequence obtained from the vertex $O_2$. In fact, $sP_2 = (1,0,0,~0^{21},~1,3,1)$, $\ell(\mbp(F)) = (0,0,1,~0,0,0,~1,3,1)$, $d_{r_1}=1$ and $d_f=3$. 

During $T_6$, robot $r_1$ is moved along the shortest side of the parallelogram associated with the string $s$ (cf. the definition of $\vpfuno$) so as to increase position $d_{r_1}$. This movement stops when $d_{r_1}=d_f$ applies and at that time task $T_7$ starts. It is easy to observe that when $T_7$ starts as a consequence of the termination of $T_6$, the following variable holds:
 \begin{itemize}
\item
     \emph{ $\vqfuno =$ sequence $\ell(\mbp(R))$ guarantees that 
                        $\ell(\mbp(R) )  = \ell’ + \LSF$, for some $\ell’$ 
                        made of only $0$'s and just one $1$ in position $d_{r_1}$ 
                        and $d_{r_1} = d_f$. }            
\end{itemize} 
Basically, when $\vqfuno$ holds, all robots know that $r_1$ can complete the pattern by going straight toward its target. The main difficulty in this task is to cope with possible symmetries formed during the last movement of $r_1$. Notice that the configurations produced by the algorithm are always asymmetric (this holds for the initial configuration, and the position of $r_1$ guarantees that this property is maintained in $T_1,\ldots,T_6$). In principle, when $r_1$ is very close to its target, possible symmetries may imply that more than one robot can detect itself as the moving robot $r_1$, and also that $r_1$ can detect more than one vertex as the target $f_1$. However, we are able to show that in any case the formed configuration has at most a reflection axis with $r_1$ on that axis. Summarizing, we will show that all these cases do not prevent the algorithm to complete the pattern formation.

%

\section{Formalization and correctness}\label{sec:formalization}
As introduced in the previous section, the proposed algorithm $\Aform$ is based on a strategy that decomposes the $\apf$ problem into tasks $T_1,T_2,\ldots,T_8$. 
According to the \LCM model, during the \Compute phase each robot must be able to recognize the task to be performed just according to the configuration perceived during the \Look phase and the input pattern $F$. This recognition can be performed by providing $\Aform$ with a \emph{predicate} $P_i$ for each task $T_i$. Given the perceived configuration and the input pattern $F$, the predicate $P_i$ that results to be true reveals to robots that the corresponding task $T_i$ is the task to be performed. This approach requires that the designed predicates must guarantee some properties:
\begin{description}
\item[$\Prop_1$:]
given the pattern $F$, each $P_i$ must be computable on the configuration $C$ perceived in each \Look phase; 
\item[$\Prop_2$:]
$P_i \wedge P_j =\false$, for each $i\neq j$; this property allows robots to exactly recognize the task to be performed;
\item[$\Prop_3$:]
given the pattern $F$, for each possible perceived configuration $C$ there must exists a predicate $P_i$ evaluated true. 
\end{description}
If we guarantee that all these properties hold, then $\Aform$ can be used in the \Compute phase as follows: 
\begin{quote}
    \emph{-- if a robot $r$ executing algorithm $\Aform$ detects that 
          predicate $P_i$ holds, then $r$ simply performs a move $m_i$
          associated with task $P_i$.}
\end{quote}
Concerning how to define the predicates, we have already remarked in the previous section that each task can be accomplished only when some \emph{pre-conditions} are fulfilled. Hence, to define the predicates in general we need:
\begin{itemize}
\item
\emph{basic variables} that capture metric/topological/numerical/ordinal aspects of the input configuration which are relevant for the used strategy and that can be evaluated by each robot on the basis of its view;
\item
\emph{composed variables} that express the pre-conditions of each task $T_i$.
\end{itemize}
All the needed basic variables useful for $\Aform$ have been already defined in Sections~\ref{ssec:T1}--\ref{ssec:T6-T7}. If we assume that $\pre_i$ is the composed variable that represents the pre-conditions of $P_i$, for each $1\le i\le 8$, then predicate $P_i$ can be defined as follow:
\begin{equation}\label{eq:predicates}
	P_i = \pre_i \wedge \neg ( \pre_{i+1} \vee \pre_{i+2} \vee \ldots \vee \pre_8 )
\end{equation} 
This definition leads to the following remark:
\begin{remark}\label{rem:prop2}
\textit{
Predicates $P_i$ fulfill Property $\Prop_2$. This is directly implied by Eq.~\ref{eq:predicates} 
} 
\end{remark}

\begin{table*}[t]
\bgroup
\def\arraystretch{1.4}
\setlength{\tabcolsep}{6pt}
\begin{center}
  \begin{tabular}{ | c | p{0.53 \textwidth} | p{0.35 \textwidth} | }
    \hline
    \textit{var}  &  \textit{definition}  
                  &\textit{rationale} 
                  \\ \hline \hline
    $\vguno$   & $\exists$ a unique line parallel to a canonical direction passing 
                 through $r_1$ and each $\bp(R')$
               & \textit{guard $r_1$ is partially placed}
               \\ \hline
    $\vgenne$  & $r_n$ is at a vertex $(0,y)$, with $2\Delta \leq y < d(r_1,O)$
               & \textit{guard $r_n$ is placed} 
               \\ \hline
    $\vdruno$  & $d(r_1,O)\ge 3\Delta$
               & \textit{guard $r_1$ is at a desired distance from the origin }
               \\ \hline
   $\vdrunop$ & $d(r_1,\mbp(F))\ge 3w(\mbp(F))$
               & \textit{robot $r_1$ is at a desired distance from the pattern } 
               \\ \hline
    $\vhpuno$  & let $L$ be the line induced by $\vguno$; all robots in $R'$ are in the 
                 same half-plane with respect to $L$
               & \textit{all robots in $R'$ are in the same half-plane with respect 
                 to the line induced by $\vguno$ } 
               \\ \hline
    $\vhpdue$  & let $L$ be the line induced by $\vguno$; all robots in $R''$ are in the 
                 same half-plane with respect to $L$
               & \textit{all robots in $R''$ are in the same half-plane with respect 
                 to the line induced by $\vguno$ } 
               \\ \hline
    $\vhrenne$ & $f_n=(x,y)$ and $r_n=(x',y')$, with $x'\leq x$ and $y'\geq y$
               & \textit{guard $r_n$ is on the right path to its target}
               \\ \hline
     $\vrpeffe$ & the largest unmatched robot $r_i$ is on a shortest path from 
                 any vertex in $Q^-$ to $f_i$, and each robot $r_j$, $j<i$, 
                 is in $Q^-$
               & \textit{all the unmatched robots are correctly positioned with respect 
                 to $\PPF$}
               \\ \hline
    $\vpfuno$  & $\exists$ $s\in \{sP_1,sP_2\}$: $s = s’ + \LSF$, for some $s’$ 
                 made of only $0$'s and just one $1$ in position $d_{r_1}$ and 
                 $d_{r_1} < d_f$
               & \textit{pattern $F \setminus \{f_1\}$ formed}
                \\ \hline
    $\vpfenne$ &  $\exists$ embedding of $F$ such that all robots in $R''$ are similar 
                  to  $F \setminus \{f_1,f_n\}$
               & \textit{pattern $F \setminus \{f_1,f_n\}$ formed} 
               \\ \hline
    $\vqfuno$  &  $\ell(\mbp(R) )  = \ell’ + \LSF$, for some $\ell’$ made of 
                  only $0$'s and just one $1$ in position $d_{r_1}$ and $d_{r_1} = d_f$
               & \textit{guard $r_1$ can complete the pattern by going straight toward 
                  its target} 
               \\ \hline
    $\vs$      & $R$ and $F$ are similar
               & \textit{pattern $F$ formed } 
               \\ \hline
  \end{tabular}
\end{center}
\egroup
\caption{ The basic Boolean variables used to define all the tasks' preconditions. }
\label{tab:basic-variables}
\end{table*}
%
Before addressing the remaining properties $\Prop_1$ and $\Prop_3$, we formalize all the basic variables, the pre-conditions for each task, and, as a consequence, all the predicates. 
All the necessary basic variables are summarized in Table~\ref{tab:basic-variables}.
Table~\ref{tab:tasks} reports all the ingredients determined by the proposed algorithm: the first two (general) columns recall the hierarchical decomposition described in the previous section, the third column associates tasks names to sub-problems, and the fourth column defines precondition $\pre_i$ for each task $T_i$. These preconditions must be considered according to Equation~\ref{eq:predicates}. The fifth column of Table~\ref{tab:tasks} contains the name of the move used in each task (we simply denote as $m_i$ the move used in task $T_i$), and the specification of each move is provided in Table~\ref{tab:moves}. Unless differently specified, each trajectory defined in the moves must be intended as any shortest path to the target. 

Table~\ref{tab:tasks} leads to the following remark:
\begin{remark}\label{rem:prop3}
\textit{
Algorithm $\Aform$ fulfills Property $\Prop_3$. This is implied by pre-condition $\pre_1$ and predicates $P_i$.
}
\end{remark}
%


\begin{table*}[t]
\bgroup
\def\arraystretch{1.4}
\setlength{\tabcolsep}{6pt}
\begin{center}
  \begin{tabular}{ | c | l | l | l | r | l | }
    \hline
\textit{problem} & \multicolumn{2}{l|}{\textit{sub-problem}} & \textit{task} &  
       \textit{precondition} & \textit{move}\\ \hline \hline

\multirow{8}{*}{\raggedleft $\apf$ } & \multirow{4}{*}{\raggedleft $\RS$ }  & $\RS_{1a}$  & \textit{ $T_1$} & \textit{true}  &  $m_1$ \\ \cline{3-6}

& & $\RS_{1b}$ & \textit{ $T_2$} & $\predue$ & $m_2$ \\ \cline{3-6}

& & $\RS_{2}$ & \textit{ $T_3$} & $\pretre$ & $m_3$ \\ \cline{2-6}
                          
& \multicolumn{2}{l|}{ $\PPF$ }  &  \textit{ $T_4$} & $\prequattro$ &  $m_4$ \\ \cline{2-6}

& \multirow{4}{*}{\raggedleft $\Fin$ } & $\Fin_{1}$ & \textit{ $T_5$} &  $\precinque$ & $m_5$ \\ \cline{3-6}

& & $\Fin_{2}$ & \textit{ $T_6$} & $\presei$ & $m_6$ \\ \cline{3-6}

& & $\Fin_{3}$ & \textit{ $T_7$} & $\presette$  & $m_7$ \\ \cline{2-6}

& \multicolumn{2}{l|}{ $\Term$ } & \textit{ $T_8$} & $\vs$  & $\nil$ \\ \hline 

\hline
  \end{tabular}
\end{center}
\egroup
\caption{ Algorithm $\Aform$ for $\apf$. }
\label{tab:tasks}
\end{table*}

\begin{table*}[t]
\bgroup
\def\arraystretch{1.3}
\setlength{\tabcolsep}{5pt}
\begin{center}
  \begin{tabular}{ | c | p{0.87 \textwidth} | }
    \hline
    \textit{move}  &  \textit{definition} \\ \hline \hline
      
    $m_1$ & $r_1$ moves toward the closest vertex so as $\vguno$ holds\\ \hline
    $m_2$ & $r_1$ moves on the closest vertex on $X$-axis at distance at least 
            $3\Delta$ from the origin \\ \hline
    $m_3$ & $r_n$ moves toward vertex $(0,y)$, with $2\Delta \leq y < d(r_1,O)$ \\ \hline
    $m_4$ & the \emph{largest unmatched robot} in $R''$ moves toward the 
            \emph{largest unmatched target} in $F_e$ \\ \hline
    $m_5$ & $r_n$ first moves along a path that maintains fixed the $y$ 
            coordinate until its $x$ coordinate coincides with that of $f_n$ - then, 
            it moves toward $f_n$ \\ \hline
    $m_6$ & if both $sP_1$ and $sP_2$ satisfy $\vpfuno$ then let $s$ be the
            lexicographically minimum one. Then, robot $r_1$ moves along the 
            shortest side of the parallelogram associated with $s$ so as to 
            increase $d_{r_1}$ \\ \hline
    $m_7$ & robot $r_1$ in position $d_{r_1}$ moves toward $f_1$ \\ \hline
  \end{tabular}
\end{center}
\egroup
\caption{ Moves associated with tasks. It is assumed that each robot not involved in $m_i$ perform the $\nil$ movement.}
\label{tab:moves}
\end{table*}


\subsection{On computing the predicates: property $\Prop_1$}\label{ssec:compute}
In this section, we show how the proposed algorithm $\Aform$ can compute each predicate $P_i$, that is, we show that $\Aform$ guarantees that property $\Prop_1$ holds. 

According to the definition of $P_i$ given in Eq.~\ref{eq:predicates}, in the \Compute phase, each robot evaluates -- with respect to the perceived configuration $C$ and the pattern $F$ to be formed -- the predicates starting from $P_8$ and proceeding in the reverse order with the others until a true pre-condition is found. In case all pre-conditions $\pre_8,\pre_7,\ldots,\pre_2$ are evaluated false, then task $P_1$ is performed.

Evaluating $\pre_8$ is just a matter of testing whether $C$ and $F$ are similar. Concerning the evaluation of $\pre_7$, robots need to compute just variable $\vqfuno$, which in turn depends only on $\mbp(R)$ and $\mbp(F)$. Pre-condition $\pre_6$ implies to compute $\vguno$,  $\vdrunop$, and $\vpfuno$: the first needs $r_1$, which can be identified according to function $D()$ since in all tasks $T_1,\ldots,T_6$ this guard corresponds to the robot $r$ such that $D(r)$ is maximum; the second just uses $r_1$ and $\mbp(F)$; the third is computable as shown in Section~\ref{ssec:T6-T7} starting from the direction $U$ induced by variable $\vguno$, and by using $U$ and $\mbp(F)$ to determine the sequences $sP_1$ and $sP_2$, and the positions $d_{r_1}$ and $d_f$.

Pre-condition $\pre_5=\precinque$ can be evaluated as follows: $\vguno$ can be evaluated once $r_1$ has been recognized thanks to function $D()$, and $\vhpuno$ can be detected once the direction induced by variable $\vguno$ is known. Now, as described in Remark~\ref{remark:T5}, by using $\vguno$ and $\vhpuno$ the common reference system can be established by each robot, and from it both $\vdruno$ and $\vhrenne$ can be evaluated. Finally, variable $\vpfenne$ can be checked by using a combinatorial approach.

Concerning pre-condition $\pre_4=\prequattro$, all its variables except $\vrpeffe$ can be evaluated according to Remark~\ref{remark:RS2}, while $\vrpeffe$ can be easily checked according to the definitions introduced in Section~\ref{ssec:T4}.
Pre-condition $\pre_3=\pretre$ can be evaluated as follows: $\vguno$ can be computed again according to $D()$, and then Remark~\ref{rem:afterT2} can be used to establish the common reference system. From this reference system, both $\vhpdue$ and $\vdruno$ can be evaluated.

Finally, for checking $\pre_2=\predue$ it is enough to use $D()$ to identify robot $r_1$.


\subsection{Correctness}\label{ssec:correctness}
In this section, we formally prove that algorithm $\Aform$ solves the $\apf$ problem on the tessellation graph $\GT$. To this end, let $\I_\A$ be the set containing all the configurations taken as input or generated by $\Aform$.

According to properties $\Prop2$ and $\Prop 3$, all tasks’ predicates $P_1$, $P_2$, $\ldots$, $P_8$ used by the algorithm have been defined so as to make a partition of $\I_\A$. Together with $\Prop_1$, for each possible configuration provided to $\Aform$, the algorithm can evaluate each predicate ad exactly determine the task to be performed. 

The correctness can be assessed by proving that all the following properties hold:

\begin{description}
\item[$H_1$:]
$\Aform$ does not generate multiplicities nor symmetric configurations (unless  $F$ is formed or its formation is not prevented);
\item[$H_2$:]
from any class $T_i$, $1\leq i\leq 8$, no class $T_j$ with $j<i$ can be reached.


\item[$H_3$:]
from any class $T_i$, $1\leq i\leq 7$, another class $T_j$ with $j>i$ is always reached within a finite number of \LCM cycles.
\end{description}

Since properties $H_1$, $H_2$ and $H_3$ must be proved for each transition/move, then in the following we provide a specific lemma for each task. 

\begin{lemma}\label{lem:T1}
From an initial configuration $C$ belonging to class $T_1\cap \I_\A$ the algorithm $\Aform$ eventually leads to a configuration $C'$ in a class $T_i$, $i>1$.
\end{lemma}

\begin{proof}
In this task, algorithm $\Aform$ selects a robot denoted as $r_1$ (the first guard) such that $D(r_1)$ is maximum and, in case of ties, the robot that has the minimum position in $\ell(\mbp(R))$.
\begin{description}
\item[$H_1$.] Since $D(r_1)$ is maximum, while $r_1$ moves away from the other robots, it cannot meet any other robot and $D(r_1)$ increases. Then, $r_1$ is repeatedly selected. Note that, if by $m_1$ a symmetric configuration is created then it must admit an axis of reflection where $r_1$ lies as this is the only robot defining $D(r_1)$.
\item[$H_2$.] Since as we are going to show the subsequent $H_3$ holds, we have that any other class can be reached.  
\item[$H_3$.] Robot $r_1$ always decreases the distance toward its target, within a finite number of \LCM cycles, unless other predicates become true, $\vguno$ becomes true and the configuration is not in $T_1$ anymore. \qed
\end{description}
\end{proof}

\begin{lemma}
From a configuration $C$ belonging to class $T_2\cap \I_\A$ the algorithm $\Aform$ eventually leads to a configuration $C'$ in a class $T_i$, $i>2$.
\end{lemma}
\begin{proof}
Here $r_1$ lies between two parallel directions $L_1$ and $L_2$ enclosing each possible $\bp(R')$ and moves toward the closest one (toward any of them in case of ties) along a canonical direction. 
 \begin{description}
\item[$H_1$.]  Robot $r_1$, when moving toward its target, cannot meet any other robot, nor move on any axis of symmetry because the only possible one should be at the same distance from $L_1$ and $L_2$ and parallel to them. However, by moving to the closest $L_i$, $i\in \{1,2\}$, $r_1$ never crosses an axis. 

\item[$H_2$.] 
Move
$m_2$ does not affect predicate $\vguno$, that is no obtained configuration can belong to $T_1$. 

\item[$H_3$.] Robot $r_1$ always decreases the distance toward $L_i$, then within a finite number of \LCM cycles, unless other predicates become true, $\vhpdue\wedge\vdruno$ becomes true (cf. Section~\ref{ssec:T2}) and the configuration is not in $T_2$ anymore. \qed
\end{description}
\end{proof}

\begin{lemma}\label{lem:T3}
From a configuration $C$ belonging to class $T_3\cap \I_\A$ the algorithm $\Aform$ eventually leads to an asymmetric configuration $C'$ in $T_i$, $i>3$.
\end{lemma}
\begin{proof}
During this task, guard $r_1$ is already placed, that is $\pretre$ holds.
 \begin{description}
\item[$H_1$.] Due to the positioning of $r_1$, the configuration can be symmetric only when all robots are collinear (along the formed $X$-axis). Regardless when this symmetry is formed, during this task, $r_n$ is always detected and as soon as it leaves the $X$-axis, the configuration becomes asymmetric and remains as such until the second guard terminates its trajectory. According to $m_3$, along its movement $r_n$ cannot meet any other robot. Notice that, according to the different distances of the two guards from $O$, the configuration cannot admit rotations nor reflections as long as the guards are idle.

\item[$H_2$.]
Move 
$m_3$ does not affect predicates $\vguno$, $\vhpdue$ and $\vdruno$, that is any obtained configuration cannot belong to $T_1$ nor to $T_2$.

\item[$H_3$.]  Robot $r_n$ always decreases the distance toward its target along the $Y$-axis, then within a finite number of \LCM cycles, unless other predicates become true, $\vgenne$ becomes true. In any case the configuration is not in $T_3$ anymore. \qed
\end{description}
\end{proof}

\begin{lemma}
From any configuration $C$ belonging to class $T_4\cap \I_\A$ the algorithm $\Aform$ eventually leads to a configuration $C'$ in $T_5$.
\end{lemma}
\begin{proof}

\begin{description}
\item[$H_1$.]  Since guards $r_1$ and $r_n$ are placed, the same considerations of Lemma~\ref{lem:T3} hold, that is the configuration cannot admit reflections nor rotations during this task. Multiplicities can be created but only if required by the formation of $F$.

\item[$H_2$.]
During the whole task, predicate $\vs$ is false as guards remain placed. Hence, also predicates $\vguno$, $\vdruno$, $\vgenne$ and $\vrpeffe$ are not affected by $m_4$, that is the obtained configuration cannot belong to $T_1$, $T_2$, and $T_3$. 

\item[$H_3$.] While the task is performed, either the number of matched robots increases or the distance of one robot from its target decreases, then in a finite number of moves all robots excluding $r_1$ and $r_n$ will be matched. As already described in Section~\ref{ssec:T4}, at the end of this task $\precinque$ holds, that is $C'$ belongs to $T_5$ and no other task can be reached because the guards remain placed.\qed
\end{description}
\end{proof}


\begin{lemma}
From any configuration $C$ belonging to class $T_5\cap \I_\A$ the algorithm $\Aform$ eventually leads to a configuration $C'$ in $T_i$, $i>5$.
\end{lemma}
\begin{proof}
During this task,  guard $r_1$ remains placed.
 \begin{description}
\item[$H_1$.] As $r_n$ moves toward its final target, the arisen configurations cannot admit reflections nor rotations as there are no other robots equivalent to $r_1$ due to $\vguno$ and $\vdruno$. A reflection (as well as a multiplicity, resp.) can occur only at the end of the task if all robots are collinear (if $f_n$ requires a multiplicity, resp.) but this can be managed by $\Aform$ as we are going to see in the next lemma devoted to $T_6$.

\item[$H_2$.]
Before $r_n$ reaches its target, $\precinque$ remains true, while predicates $\vpfuno$ and $\vs$ remain false, that is  the configuration remains in $T_5$. 
Once $r_n$ reaches $f_n$, predicate $\presei$ becomes true.

\item[$H_3$.] After each move, $r_n$ decreases its distance  from $f_n$, that is within a finite number of $\LCM$ cycles the task ends and, unless other predicates become true, the obtained configuration $C'$ belongs to $T_6$. \qed
\end{description}
\end{proof}

\begin{lemma}\label{lem:T6}
From a configuration $C$ belonging to class $T_6\cap \I_\A$ the algorithm $\Aform$ eventually leads to a configuration $C'$ in $T_7$.
\end{lemma}


\begin{figure}[t]
   \graphicspath{{fig/}}
   \centering
   \def\svgwidth{\columnwidth}
   {\large \scalebox{0.7}{
\begingroup%
  \makeatletter%
  \providecommand\color[2][]{%
    \errmessage{(Inkscape) Color is used for the text in Inkscape, but the package 'color.sty' is not loaded}%
    \renewcommand\color[2][]{}%
  }%
  \providecommand\transparent[1]{%
    \errmessage{(Inkscape) Transparency is used (non-zero) for the text in Inkscape, but the package 'transparent.sty' is not loaded}%
    \renewcommand\transparent[1]{}%
  }%
  \providecommand\rotatebox[2]{#2}%
  \newcommand*\fsize{\dimexpr\f@size pt\relax}%
  \newcommand*\lineheight[1]{\fontsize{\fsize}{#1\fsize}\selectfont}%
  \ifx\svgwidth\undefined%
    \setlength{\unitlength}{447.54366295bp}%
    \ifx\svgscale\undefined%
      \relax%
    \else%
      \setlength{\unitlength}{\unitlength * \real{\svgscale}}%
    \fi%
  \else%
    \setlength{\unitlength}{\svgwidth}%
  \fi%
  \global\let\svgwidth\undefined%
  \global\let\svgscale\undefined%
  \makeatother%
  \begin{picture}(1,0.36581271)%
    \lineheight{1}%
    \setlength\tabcolsep{0pt}%
    \put(0,0){\includegraphics[width=\unitlength,page=1]{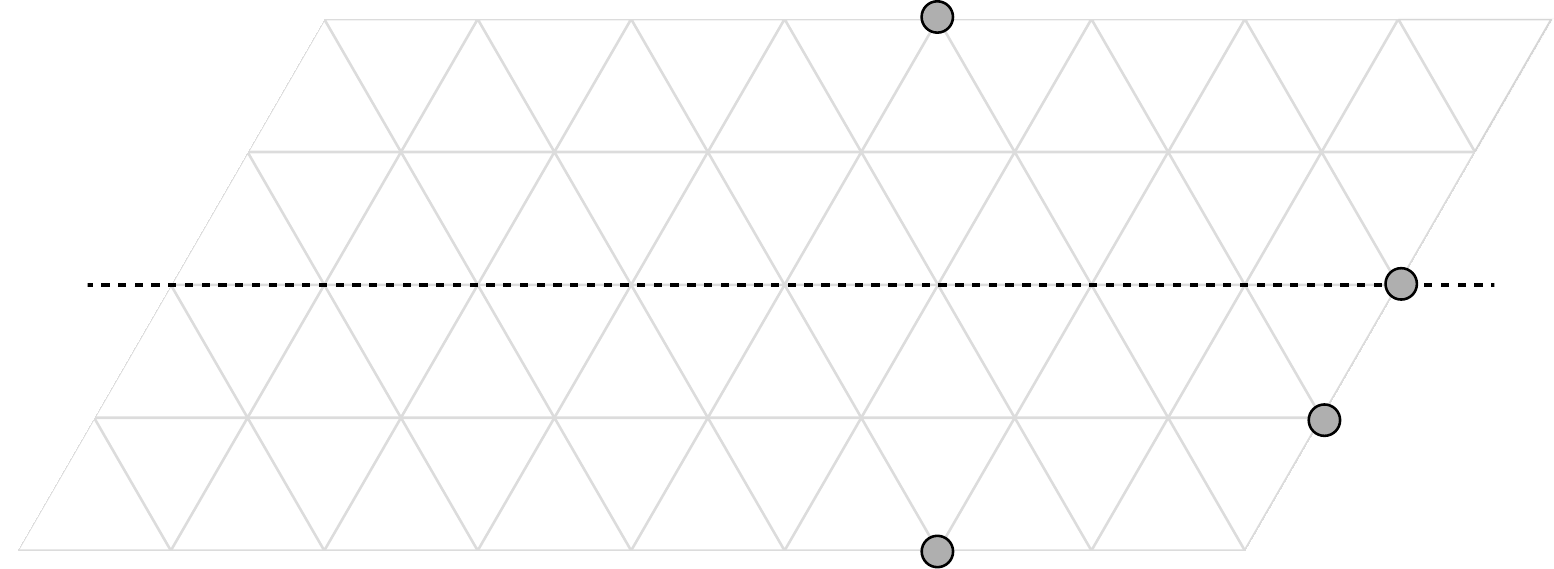}}%
    \put(0.03183523,0.01157944){\color[rgb]{0,0,0}\makebox(0,0)[lt]{\lineheight{1.25}\smash{\begin{tabular}[t]{l}$r_1$\end{tabular}}}}%
    \put(0.57504652,0.082061){\color[rgb]{0,0,0}\makebox(0,0)[lt]{\lineheight{1.25}\smash{\begin{tabular}[t]{l}$f_1'$\end{tabular}}}}%
    \put(0,0){\includegraphics[width=\unitlength,page=2]{INKT6-proof.pdf}}%
    \put(0.58138586,0.25862446){\color[rgb]{0,0,0}\makebox(0,0)[lt]{\lineheight{1.25}\smash{\begin{tabular}[t]{l}$f_1$\end{tabular}}}}%
    \put(0,0){\includegraphics[width=\unitlength,page=3]{INKT6-proof.pdf}}%
    \put(0.06985707,0.20121979){\color[rgb]{0,0,0}\transparent{0.41000003}\makebox(0,0)[lt]{\lineheight{1.25}\smash{\begin{tabular}[t]{l}$L$\end{tabular}}}}%
  \end{picture}%
\endgroup%
}}
  \caption{%
An example of the only possible symmetry that can arise during task $T_6$. 
}
\label{fig:T6_proof}
\end{figure}

\begin{proof}
During this task, guard $r_1$ moves so as to make $d_{r_1}=d_{f_1}$.
 \begin{description}
\item[$H_1$.]
During this phase the algorithm does not generate any multiplicity since $\presei$ remains true and then $r_1$ is sufficiently far from any other robot.
Regarding symmetries (cf.~Figure~\ref{fig:T6_proof}), the only symmetric configuration possible is the one with an axis parallel to the direction induced by $\vguno$, and $f_1$ can be on the axis or not. In the first case, the whole pattern is symmetric; when $r_1$ reaches the axis then $d_{r_1}=d_{f_1}$ holds and predicate $\presette$ becomes true. Otherwise the final pattern is asymmetric and there are two possible embeddings and two possible targets for $r_1$, $f_1$ and its equivalent point $f_1'$ with respect to the axis of symmetry.
One of the two is reachable by $r_1$ without crossing the axis. Targets $f_1$ and $f_1’$ lie in the same half plane or not. In the first case only one among the sequences $sP_1$ and $sP_2$ satisfies the condition in $\vpfuno$ because in one of them $d_{r_1}>d_{f_1}$. Then $r_1$ moves towards $f_1$ and when it reaches the height of $f_1’$ predicate $\presette$ becomes true and the configuration is in $T_7$. 
When $r_1$ lies between $f1$ and $f_1’$, both the sequences $sP_1$ and $sP_2$ satisfy the condition in $\vpfuno$ and move $m_6$ chooses the smaller one since $sP_1$ and $sP_2$ must be different because $r_1$ is not on the axis.
 Robot $r_1$ increases its height to align with the target until $d_{r_1}=d_{f_1}$ and the configuration is in $T_7$.

\item[$H_2$.]
Before $r_1$ reaches the height of its target, $\presei$ remains true, and when $d_{r_1}=d_{f_1}$, $\presette$ holds, hence $C'$ is in $T_7$. Clearly $C'$ cannot belong to $T_8$. 

\item[$H_3$.]
The absolute difference between $d_r$ and $d_f$ decreases by one at each move until $d_{r_1}=d_{f_1}$ so that move $m_6$ is applied only a finite number of times. \qed
\end{description}
\end{proof}

\begin{lemma}\label{lem:T7}
From a configuration $C$ belonging to class $T_7\cap \I_\A$ the algorithm $\Aform$ eventually leads to a configuration $C'$ in $T_8$.
\end{lemma} 

\begin{proof}

During this task, guard $r_1$ straightly moves toward its target. Since $\vqfuno$ holds it is possible to derive the embedding of the pattern from $\ell(\mbp(F))$ and consequently the $X$ and $Y$ axis that we refer to the proof. 
 \begin{description}
\item[$H_1$.] We show that while $r_1$ moves toward $f_1$ no reflections, no rotations, no multiplicities can be created that prevent the finalization of the task. In particular, we first show that no rotations are possible, then we analyze reflections showing that none of them can admit a robot equivalent to $r_1$. Hence, if a reflection is created, then $r_1$ must be on the axis of symmetry, and we show this happens only if $F\setminus\{f_1\}$ is symmetric respect to that axis. Regarding the multiplicities, $r_1$ can make one only once $f_1$ is reached.
\medskip

\emph{Rotations.} The minimal possible angle of rotation is $60^\circ$ and its multiples $120^\circ$ and $180^\circ$, clockwise and anti-clockwise.
The convex hull of any configuration with rotational symmetry with angle of rotation of $60^\circ$ is an hexagon. Assuming that such a configuration is formed when $r_1$ is approaching its target, a part of the convex hull should be in the quadrant where $r_1$ lies. Then the embedding of the pattern is not positioned according to the rule that the shorter side of the parallelogram is parallel to the $Y$-axis (cf Definition~\ref{def:embedding}). With the same arguments we can exclude rotations of $120^\circ$. Regarding to rotations of $180^\circ$, let us assume that $r_1$ creates such a symmetry when approaching its target. The embedding $F_e$ is done by construction in such a way that the sequence of integers read from the origin is smaller than the one read from the corner $P=(x_p,y_p)$ at the opposite angle of $60^\circ$. The first column read from $P$ must have a single robot $r'_1$, symmetric to $r_1$, because this column matches the one with $r_1$. 

By hypothesis, the pattern sequence read from $O$ must be lower than the one read from $P$, then the first column cannot have more than one target and in particular this target must be at the same distance from $O$ than $r'_1$ from $P$, because $r_1$ is moving horizontally. Reading the configuration forward from $P$, there must be a sequence of columns of zeros, at least one, each corresponding to an empty column read from $r_1$ to the $Y$-axis, that is empty. In turn, this corresponds to a sequence of columns of zeros in the pattern read from $O$, because by hypothesis must be lower than the one read from $P$. Then, by rotation, these columns correspond to more empty columns in the configuration read from $P$. Continuing, we would have only empty columns between $r_1$ and $r'_1$, contradicting the hypothesis that the robots are at least three.
\medskip

\emph{Reflections.}
Regarding reflections  we have to analyze possible axis of reflection at $0^\circ$, $30^\circ$, $60^\circ$, $90^\circ$, $120^\circ$, $150^\circ$, with respect to the $X$-axis in clockwise direction. Moreover we distinguish between two cases: when $r_1$ becomes equivalent to another robot of the configuration and when $r_1$ goes on an axis of symmetry.
\begin{figure}[t]
    \graphicspath{{fig/}}
    \centering
    \def\svgwidth{\columnwidth}
    {\large \scalebox{0.65}{
\begingroup%
  \makeatletter%
  \providecommand\color[2][]{%
    \errmessage{(Inkscape) Color is used for the text in Inkscape, but the package 'color.sty' is not loaded}%
    \renewcommand\color[2][]{}%
  }%
  \providecommand\transparent[1]{%
    \errmessage{(Inkscape) Transparency is used (non-zero) for the text in Inkscape, but the package 'transparent.sty' is not loaded}%
    \renewcommand\transparent[1]{}%
  }%
  \providecommand\rotatebox[2]{#2}%
  \newcommand*\fsize{\dimexpr\f@size pt\relax}%
  \newcommand*\lineheight[1]{\fontsize{\fsize}{#1\fsize}\selectfont}%
  \ifx\svgwidth\undefined%
    \setlength{\unitlength}{534.84049398bp}%
    \ifx\svgscale\undefined%
      \relax%
    \else%
      \setlength{\unitlength}{\unitlength * \real{\svgscale}}%
    \fi%
  \else%
    \setlength{\unitlength}{\svgwidth}%
  \fi%
  \global\let\svgwidth\undefined%
  \global\let\svgscale\undefined%
  \makeatother%
  \begin{picture}(1,0.44842438)%
    \lineheight{1}%
    \setlength\tabcolsep{0pt}%
    \put(0,0){\includegraphics[width=\unitlength,page=1]{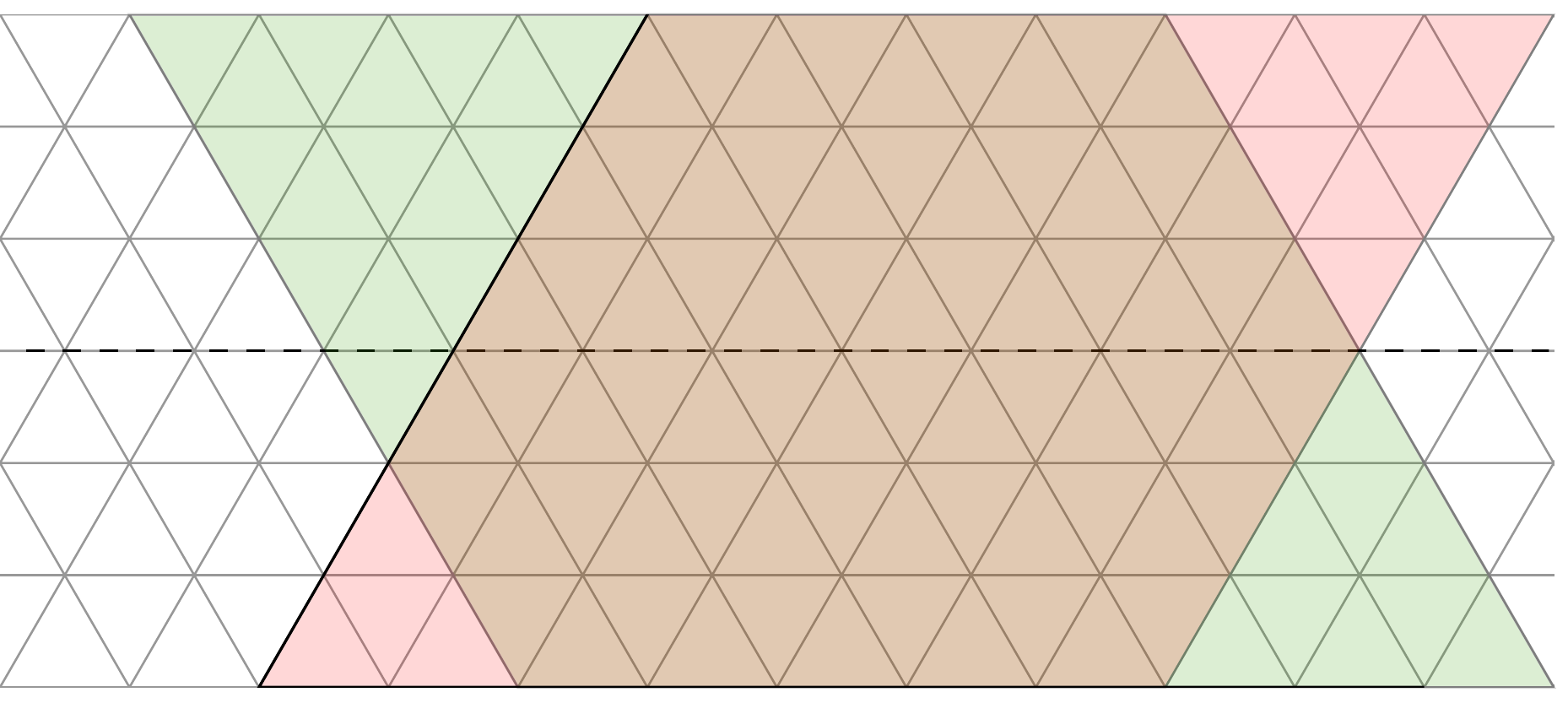}}%
    \put(0.28901791,0.38628456){\color[rgb]{0,0,0}\makebox(0,0)[lt]{\lineheight{1.25}\smash{\begin{tabular}[t]{l}$r_1$\end{tabular}}}}%
    \put(0.30635496,0.05476921){\color[rgb]{0,0,0}\makebox(0,0)[lt]{\lineheight{1.25}\smash{\begin{tabular}[t]{l}$r_1'$\end{tabular}}}}%
    \put(0,0){\includegraphics[width=\unitlength,page=2]{asseZero.pdf}}%
    \put(0.39523825,0.36303041){\color[rgb]{0,0,0}\makebox(0,0)[lt]{\lineheight{1.25}\smash{\begin{tabular}[t]{l}$f_1$\end{tabular}}}}%
    \put(0.19109751,0.01663741){\color[rgb]{0,0,0}\makebox(0,0)[lt]{\lineheight{1.25}\smash{\begin{tabular}[t]{l}$bp(F)$\end{tabular}}}}%
    \put(0.9122055,0.01959682){\color[rgb]{0,0,0}\makebox(0,0)[lt]{\lineheight{1.25}\smash{\begin{tabular}[t]{l}$bp'(F)$\end{tabular}}}}%
    \put(0,0){\includegraphics[width=\unitlength,page=3]{asseZero.pdf}}%
    \put(0.13860248,0.00199388){\color[rgb]{0,0,0}\makebox(0,0)[lt]{\lineheight{1.25}\smash{\begin{tabular}[t]{l}$O$\end{tabular}}}}%
  \end{picture}%
\endgroup%
}}
   	\caption{
An example in which $r_1$ becomes equivalent to another robot $r'_1$ respect to an axis of $0^\circ$ while moving toward $f_1$.
}
\label{fig:asseZero}
\end{figure}

Firstly, we analyze the case of a reflection at $0^\circ$ when  $r_1$ becomes equivalent to another robot $r'_1$ while moving toward $f_1$. Now consider the other possible $\bp'(F)$ having two sides parallel to the $X$-axis and shared with the chosen $\bp(F)$. One side of $\bp'(F)$ passes through $r'_1$ and the reading  from this side is lower than the reading of $\bp(F)$ from the origin. Then the embedding chosen was not coherent with the definition, a contradiction.  

For the cases of reflections at $30^\circ$ and $60^\circ$ the supposed robot $r'_1$ equivalent to $r_1$, would lie outside the embedding of $\mbp(F)$.
\begin{figure}[t]
    \graphicspath{{fig/}}
    \centering
    \def\svgwidth{\columnwidth}
    {\large \scalebox{0.65}{
\begingroup%
  \makeatletter%
  \providecommand\color[2][]{%
    \errmessage{(Inkscape) Color is used for the text in Inkscape, but the package 'color.sty' is not loaded}%
    \renewcommand\color[2][]{}%
  }%
  \providecommand\transparent[1]{%
    \errmessage{(Inkscape) Transparency is used (non-zero) for the text in Inkscape, but the package 'transparent.sty' is not loaded}%
    \renewcommand\transparent[1]{}%
  }%
  \providecommand\rotatebox[2]{#2}%
  \newcommand*\fsize{\dimexpr\f@size pt\relax}%
  \newcommand*\lineheight[1]{\fontsize{\fsize}{#1\fsize}\selectfont}%
  \ifx\svgwidth\undefined%
    \setlength{\unitlength}{441.8775135bp}%
    \ifx\svgscale\undefined%
      \relax%
    \else%
      \setlength{\unitlength}{\unitlength * \real{\svgscale}}%
    \fi%
  \else%
    \setlength{\unitlength}{\svgwidth}%
  \fi%
  \global\let\svgwidth\undefined%
  \global\let\svgscale\undefined%
  \makeatother%
  \begin{picture}(1,0.46724591)%
    \lineheight{1}%
    \setlength\tabcolsep{0pt}%
    \put(0,0){\includegraphics[width=\unitlength,page=1]{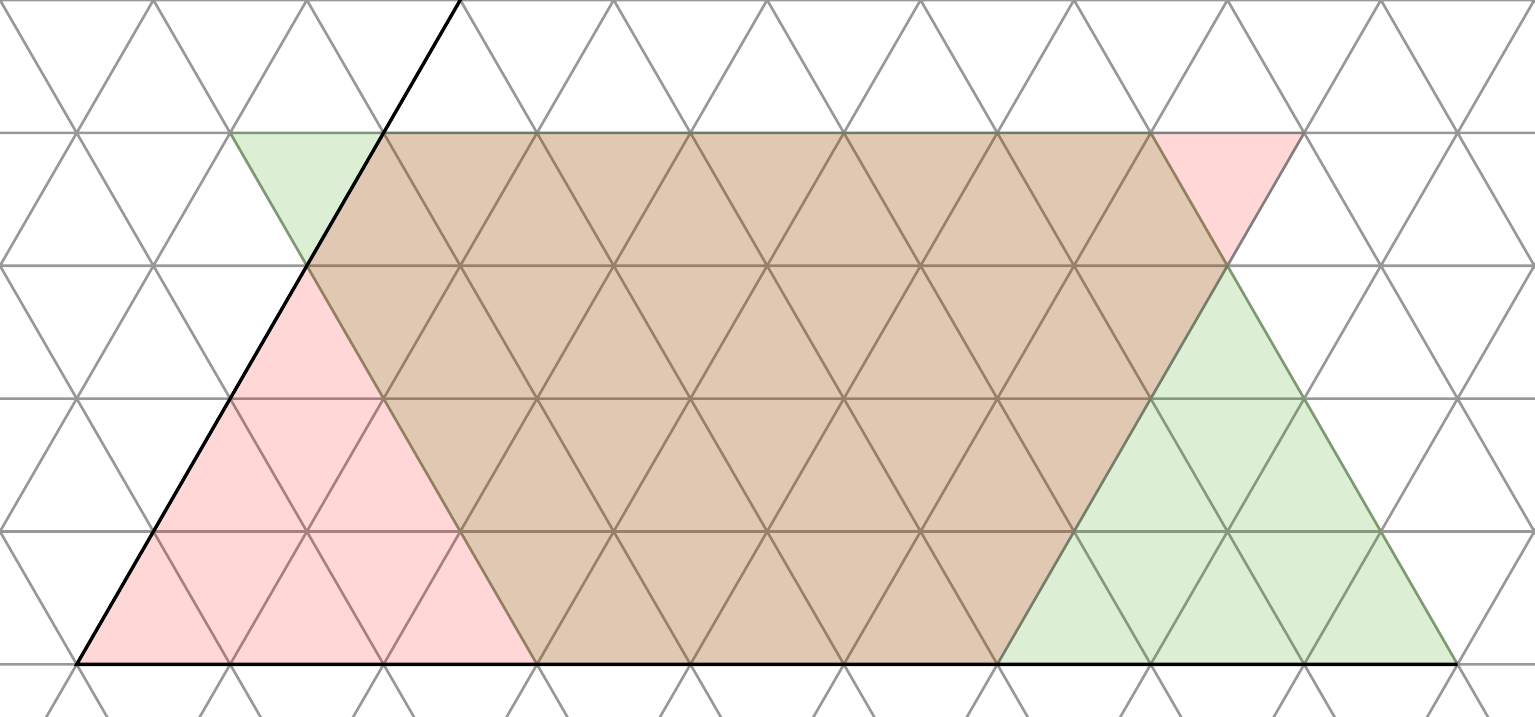}}%
    \put(0.0541537,0.26773516){\color[rgb]{0,0,0}\makebox(0,0)[lt]{\lineheight{1.25}\smash{\begin{tabular}[t]{l}$r_1$\end{tabular}}}}%
    \put(0,0){\includegraphics[width=\unitlength,page=2]{asseNovanta.pdf}}%
    \put(0.82497183,0.26407787){\color[rgb]{0,0,0}\makebox(0,0)[lt]{\lineheight{1.25}\smash{\begin{tabular}[t]{l}$r_1'$\end{tabular}}}}%
    \put(0,0){\includegraphics[width=\unitlength,page=3]{asseNovanta.pdf}}%
    \put(0.21854093,0.30287094){\color[rgb]{0,0,0}\makebox(0,0)[lt]{\lineheight{1.25}\smash{\begin{tabular}[t]{l}$f_1$\end{tabular}}}}%
    \put(0.07174893,0.04598753){\color[rgb]{0,0,0}\makebox(0,0)[lt]{\lineheight{1.25}\smash{\begin{tabular}[t]{l}$bp(F)$\end{tabular}}}}%
    \put(0.86934269,0.04240552){\color[rgb]{0,0,0}\makebox(0,0)[lt]{\lineheight{1.25}\smash{\begin{tabular}[t]{l}$bp'(F)$\end{tabular}}}}%
    \put(0.01328216,0.0378153){\color[rgb]{0,0,0}\makebox(0,0)[lt]{\lineheight{1.25}\smash{\begin{tabular}[t]{l}$O$\end{tabular}}}}%
    \put(0,0){\includegraphics[width=\unitlength,page=4]{asseNovanta.pdf}}%
  \end{picture}%
\endgroup%
}}
    \caption{%
An example in which $r_1$ becomes equivalent to another robot $r'_1$  respect to an axis of $90^\circ$ while moving toward $f_1$.
}
\label{fig:asseNovanta}
\end{figure}

Regarding the case of a reflection at  $90^\circ$, that is a reflection perpendicular to the $X$-axis,  $r_1$ becomes equivalent to another robot $r'_1$ while moving toward its target. Now consider the other possible $\bp'(F)$ having two sides parallel to the $X$-axis and shared with the chosen $\bp(F)$. As in the case of a reflection at $0^\circ$, one side of $\bp'(F)$ passes through $r'_1$ and the reading  from this side is lower than the reading of $\bp(F)$ from the origin. Then the embedding chosen was not coherent with the definition, a contradiction. 

Regarding the case of a reflection at  $120^\circ$, the reflectional axis is parallel to the $Y$-axis, and $r_1$ becomes equivalent to another robot $r'_1$ while moving toward its target. 
The axis of symmetry must be between $O$ and the half of the longest side of $\bp(F)$. We now compare the reading of $\bp(F)$ from $O$ with the reading of $\bp(F)$ starting from the corner at the opposite angle of $60^\circ$ respect to $O$, call it $P$. The first column read from $P$ has at most one robot, $r'_1$ equivalent to $r_1$, then as many empty columns as those found from $r_1$ to the $Y$-axis in $mpb(R)$,until a first robot specular to the one read from $O$. Since the number of empty columns read from $P$ is greater than the one read from $O$, the reading from $P$ is lower than the reading from $O$ hence a contradiction.

In case of a reflection axis at  $150^\circ$, the $Y$-axis reflects on the $X$-axis, then there is no possible robot $r'_1$ in the configuration that can be equivalent to $r_1$ when approaching to its target. 

In what follows, we analyze the case when $R\setminus\{r_1\}$ forms an axis of symmetry.

Consider the case of a reflection at $0^\circ$. If the pattern is symmetric respect to that axis, $f_1$ is on the axis, and $r_1$ reaches the axis and proceeds along the axis without breaking the symmetry, by following the trajectory specified by move $m_7$. If the pattern is asymmetric, then there are two possible embedding of $F$ on $R\setminus\{r_1\}$ and then there must be another target $f'_1$ equivalent to $f_1$ obtained by reflecting the embedding such that the trajectory computed by the move of $r_1$ does not cross the axis (see Lemma~\ref{lem:T6}). According to move $m_7$, actually robot $r_1$ moves to $f'_1$ to finalize the task. 

\begin{figure}[t]
   \graphicspath{{fig/}}
   \centering
   \def\svgwidth{\columnwidth}
   {\large \scalebox{0.5}{
\begingroup%
  \makeatletter%
  \providecommand\color[2][]{%
    \errmessage{(Inkscape) Color is used for the text in Inkscape, but the package 'color.sty' is not loaded}%
    \renewcommand\color[2][]{}%
  }%
  \providecommand\transparent[1]{%
    \errmessage{(Inkscape) Transparency is used (non-zero) for the text in Inkscape, but the package 'transparent.sty' is not loaded}%
    \renewcommand\transparent[1]{}%
  }%
  \providecommand\rotatebox[2]{#2}%
  \newcommand*\fsize{\dimexpr\f@size pt\relax}%
  \newcommand*\lineheight[1]{\fontsize{\fsize}{#1\fsize}\selectfont}%
  \ifx\svgwidth\undefined%
    \setlength{\unitlength}{443.17656463bp}%
    \ifx\svgscale\undefined%
      \relax%
    \else%
      \setlength{\unitlength}{\unitlength * \real{\svgscale}}%
    \fi%
  \else%
    \setlength{\unitlength}{\svgwidth}%
  \fi%
  \global\let\svgwidth\undefined%
  \global\let\svgscale\undefined%
  \makeatother%
  \begin{picture}(1,1.04085353)%
    \lineheight{1}%
    \setlength\tabcolsep{0pt}%
    \put(0,0){\includegraphics[width=\unitlength,page=1]{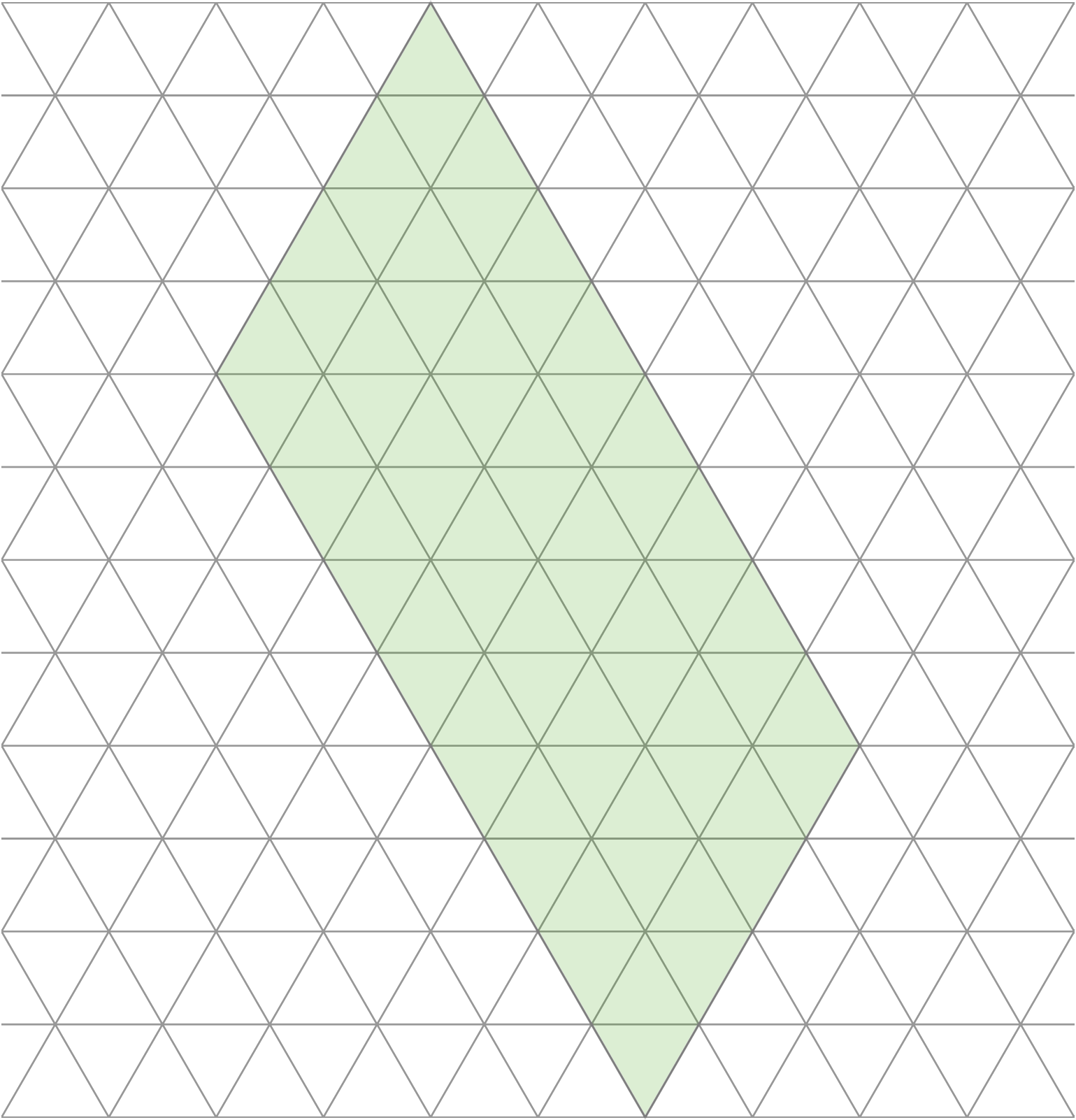}}%
    \put(0.03599016,0.35844401){\color[rgb]{0,0,0}\makebox(0,0)[lt]{\lineheight{1.25}\smash{\begin{tabular}[t]{l}$P$\end{tabular}}}}%
    \put(0,0){\includegraphics[width=\unitlength,page=2]{asseTrenta.pdf}}%
    \put(0.19549686,0.72221673){\color[rgb]{0,0,0}\makebox(0,0)[lt]{\lineheight{1.25}\smash{\begin{tabular}[t]{l}$r_1$\end{tabular}}}}%
    \put(0,0){\includegraphics[width=\unitlength,page=3]{asseTrenta.pdf}}%
    \put(0.42794377,0.70152072){\color[rgb]{0,0,0}\makebox(0,0)[lt]{\lineheight{1.25}\smash{\begin{tabular}[t]{l}$f_1$\end{tabular}}}}%
    \put(0.32138226,0.50404344){\color[rgb]{0,0,0}\makebox(0,0)[lt]{\lineheight{1.25}\smash{\begin{tabular}[t]{l}$f_1'$\end{tabular}}}}%
    \put(0.38258327,0.96288476){\color[rgb]{0,0,0}\makebox(0,0)[lt]{\lineheight{1.25}\smash{\begin{tabular}[t]{l}$P'$\end{tabular}}}}%
  \end{picture}%
\endgroup%
}}
   \caption{Robot $r_1$ on a reflection axis of $30^\circ$ and the equivalent parallelograms $P$ and $P'$.}
\end{figure}

In case of a reflection axis of $30^\circ$, $r_1$ goes towards that axis and when it lands on it there two equivalent parallelograms 
$P=\mbp(R)$ and its reflection $P'$. Let $l(P)$ and $l(P')$ the readings of the two parallelograms. These sequences are equivalent and they both find $r_1$ as the first robot. In each sequence $r_1$ is univocally determined and it can move respect to either $P$ or $P'$ toward $f_1$ or $f_1'$, respectively.  As $r_1$ moves away from the axis, there is a unique $\mbp(R)$ until $r_1$ reaches its target. 

It is easy to see that when moving $r_1$ cannot go on an axis of $60^\circ$, $90^\circ$, and $120^\circ$ before reaching its target.

Regarding to axes of $150^\circ$, $r_1$ could go on such an axis only if $f_1$ is under the reflection axis, but to be symmetric with such an axis the pattern should have the longest side laying on the $Y$-axis and this is not coherent with the embedding.

In conclusion, when moving $r_1$ does not create any rotation or reflection with a robot becoming equivalent to $r_1$. The two cases in which $r_1$ creates a symmetric configuration is when it is on an horizontal axis and it moves along that axis or when is on a $30^\circ$ axis and in this situation $r_1$ can always break the symmetry.

In order to conclude the proof of $H_1$, we also need to ensure that $r_1$ is always recognized until reaching $f_1$. In fact, as long as $r_1$ is sufficiently far away from the other robots it is easily recognizable according to its distance from $O$. When $r_1$ is close to the other robots is still always recognizable. In fact the parallelogram $\mbp(R)$ is unique (apart from the case in which $r_1$ is on an axis of symmetry at $0^\circ$ and $30^\circ$) and it can't be a square due the position of $r_1$ then there are two sequences of integers associated to the canonical corners of the $\mbp(R)$. The minimal one finds $r_1$ as the first robot; in fact if there were another robot playing the role of $r_1$ in the minimal reading that reading would be a palindrome to the first sequence and that means that the configuration is symmetric. Since the algorithm doesn't create symmetric configurations, such palindrome reading cannot exists and then $r_1$ is unique. 
If $r_1$ lies on an axis, there are two parallelograms equivalent to $\mbp(R)$ but the sequence of integers associated with these parallelograms finds $r_1$ as the first robot, then again $r_1$ is uniquely identified.

\item[$H_2$.] 
During the movement of $r_1$, predicate $\presette$ remains true because $n-1$ robots are already matched, they all stay still and $r_1$ straightly moves towards its target along the direction of the longest side of $\mbp(F)$. This implies that the sequence $\ell(\mbp(R))$ keeps its structure given by the concatenation of a subsequence $\ell’$ made of only 0s and just one 1 in position $d_{r_1}$ and a subsequence $\LSF$ that encodes the position of the robots already matched. When $r_1$ reaches its target $\ell(\mbp(R))=\ell(\mbp(F))$ and the configuration is in $T_8$.

\item[$H_3$.] After each move, $r_1$ decreases the distance from $f_1$ while the sequence $\ell’$ gets smaller by a number of 0s equal to the shorter side of $\mbp(R)$ until  $\ell(\mbp(R))=\ell(\mbp(F))$. This implies that within a finite number of $\LCM$ cycles $\vs$ becomes true and $C'$ belongs to $T_8$. \qed
\end{description}
\end{proof}

\begin{remark}
We have shown that in fact algorithm $\Aform$ manages not only asymmetric configurations but also some leader configurations where only one robot has to move and it is recognizable as one of the two guards $r_1$ or $r_n$.
\end{remark}

\begin{theorem}[Correctness]
Let $C = (\GT,\lambda)$ be any initial configuration with $n\ge 3$ \async robots, and let $F$ be any pattern (possibly with multiplicities) such that $|F|=n$. Then, $\Aform$ is able to form $F$ starting from $C$.
\end{theorem}
\begin{proof}
What we are going to show is that if all properties $H_1,\ldots,H_3$ hold, then for each possible execution of $\Aform$ there exists a time $t^*$ such that $C(t^*)$ is similar to $F$ and $C(t)=C(t^*)$ for any time $t\ge t^*$. This implies that the statement holds. 

Assume that $C$ is provided as input to $\Aform$. According to properties $\Prop_1,\ldots\Prop_3$, there exists a single task (say $T_i$) to be assigned to robots with respect to $C$. According to $H_1$, any configuration generated from $T_i$ (say $C'$) can be provided as input to $\Aform$.  Moreover, by $H_2$ and $H_3$, we can consider $C'$ belonging to some class (say $T_j$) different from $T_i$. 
According to this analysis, we can say that $C'$ will evolve during the time by changing its membership from class to class according to the forward transitions defined by Lemmas~\ref{lem:T1}--\ref{lem:T6}. 
Although the execution of $\Aform$ is infinite, property $H_3$ assures that any task is completed within a finite number of \LCM cycles, apart for $T_8$ that will be reached within finite time $t^*$. Moreover, as the only movement allowed in $T_8$ is the $\nil$ one, then the reached configuration will not change anymore. 
\qed
\end{proof}

\section{Extending the algorithm to graphs $\GS$ and $\GH$}\label{sec:extensions}

In this section, we briefly discuss how algorithm $\Aform$ can be extended to solve the $\apf$ problem for asymmetric configurations defined on $\GS$ or $\GH$. 

The proposed algorithm uses few geometric concepts, such as: bounding parallelogram, grid line, shortest path, moving along a line, quadrant. Moving from $\GT$ to $\GS$ all these concepts remain valid, with the simplification that the canonical directions are reduced to two and consequently $\bp(R)$ is unique. 
Moreover the moves do not need any changes and since predicates are independent from the underlying graph there is no need to change them. Hence the algorithm $\Aform$ remains the same and the proof of its correctness still hold taking into consideration the necessary variations needed due the reduction of the canonical directions. 

Moving to hexagonal grids, $\GH$ is considered as a sub graph of $\GT$ in which the center of the hexagons correspond to removed vertices. However by simply assuming the ``presence'' of the missing nodes and edges with respect to $\GT$, most of the geometric concepts introduced are still valid with the exception of ``movement along a line''. 
In fact it cannot move along a line but it needs to move along the edges of successive hexagons. 
For instance, in tasks $T_1$ and $T_2$, $\Aform$ simply requires that $r_1$ reaches the target via shortest paths, without assuming other constraints. So, even in $\GH$ the moves $m_1$ and $m_2$ remain valid. Conversely, during $T_3$ $r_n$ moves along the $Y$-axis according to $\Aform$. In this case we need to specify how the move unfolds since there are missing edges respect to $\GT$. In the following paragraph we revise the algorithm and give the details of the changes needed in order to extend $\Aform$ for hexagonal grids.

\subsection{Hexagonal grid graphs}

$\GH$ is considered as a sub graph of $\GT$ in which the center of the hexagons correspond to removed vertices. The basics concepts defined for $\GT$ naturally extend to $\GH$.  In particular:
\begin{itemize}
\item
the distance function between two vertices u and v in $\GH$ is the length of a shortest path connecting u and v in $\GT$;
\item
canonical directions in $\GH$ are the directions of the edges incident to a single vertex, the same introduced in $\GT$.
Given the canonical directions, we consider the same definition for an $mpb$ as in $\GT$.
Given a vertex v and oriented line L passing through v toward a canonical direction, vertex v can be classified in one of these three types:
\begin{itemize}
\item 
type 0: if v is not in $\GH$;
\item
type 1: if v has an edge following the orientation of L;
\item 
type 2 otherwise. 
\end{itemize}
The type of a leading corner is determined by the reading in the same direction that originates the sequence of $\mbp(F)$.
\item
the sequence of integers associated to a configuration of robots is the same as defined for $\GT$ placing a zero in the sequence in correspondence of a vertex in $\GT$ but not in $\GH$.
\end{itemize}

\begin{figure}[t]
   \graphicspath{{fig/}}
   \centering
   \def\svgwidth{0.6\columnwidth}
   {\large \scalebox{0.8}{
\begingroup%
  \makeatletter%
  \providecommand\color[2][]{%
    \errmessage{(Inkscape) Color is used for the text in Inkscape, but the package 'color.sty' is not loaded}%
    \renewcommand\color[2][]{}%
  }%
  \providecommand\transparent[1]{%
    \errmessage{(Inkscape) Transparency is used (non-zero) for the text in Inkscape, but the package 'transparent.sty' is not loaded}%
    \renewcommand\transparent[1]{}%
  }%
  \providecommand\rotatebox[2]{#2}%
  \newcommand*\fsize{\dimexpr\f@size pt\relax}%
  \newcommand*\lineheight[1]{\fontsize{\fsize}{#1\fsize}\selectfont}%
  \ifx\svgwidth\undefined%
    \setlength{\unitlength}{419.78366872bp}%
    \ifx\svgscale\undefined%
      \relax%
    \else%
      \setlength{\unitlength}{\unitlength * \real{\svgscale}}%
    \fi%
  \else%
    \setlength{\unitlength}{\svgwidth}%
  \fi%
  \global\let\svgwidth\undefined%
  \global\let\svgscale\undefined%
  \makeatother%
  \begin{picture}(1,0.65224019)%
    \lineheight{1}%
    \setlength\tabcolsep{0pt}%
    \put(0,0){\includegraphics[width=\unitlength,page=1]{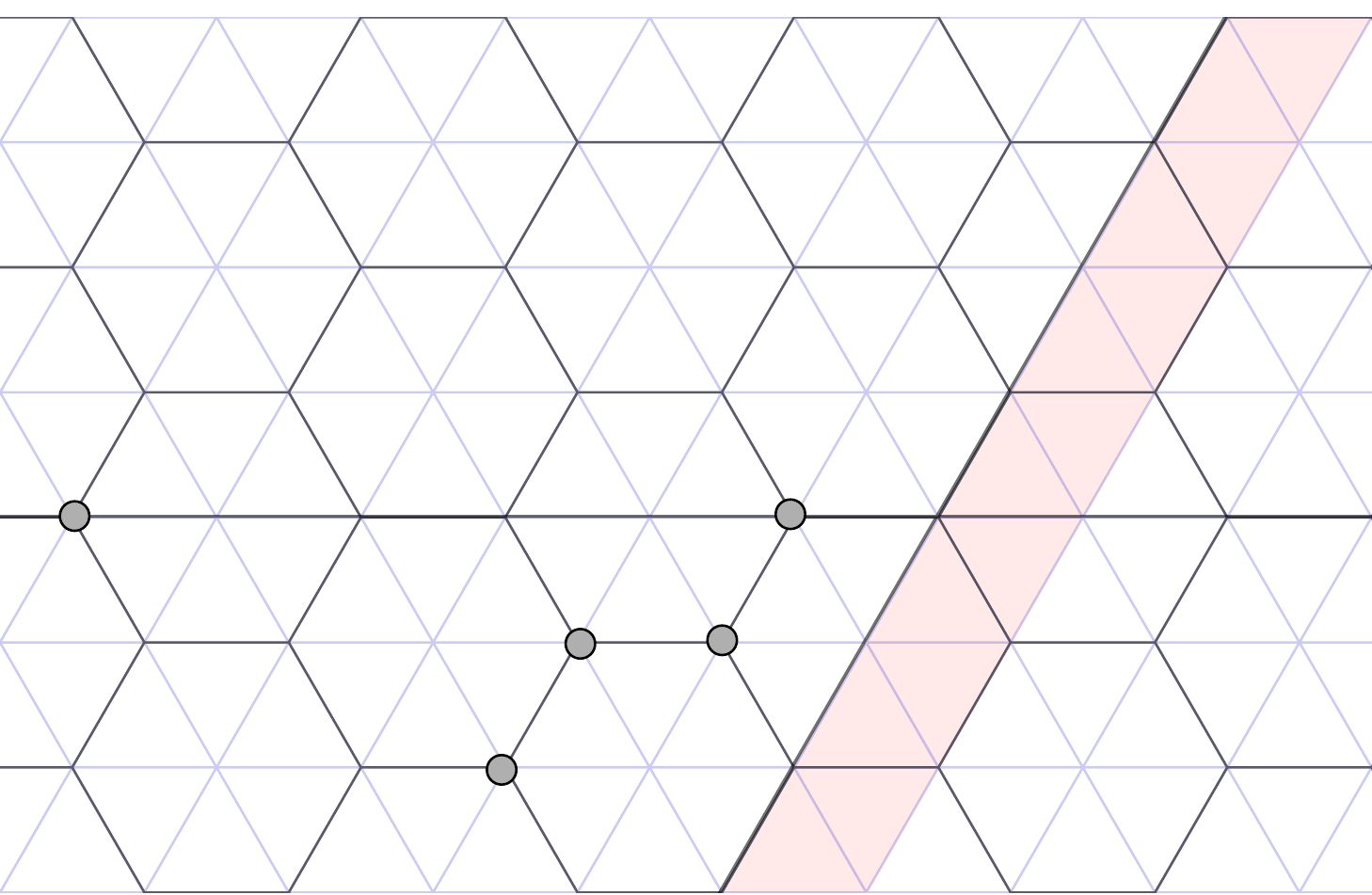}}%
    \put(0.07939354,0.28198267){\color[rgb]{0,0,0}\makebox(0,0)[lt]{\lineheight{1.25}\smash{\begin{tabular}[t]{l}$r_1$\end{tabular}}}}%
    \put(0.5895517,0.28593944){\color[rgb]{0,0,0}\makebox(0,0)[lt]{\lineheight{1.25}\smash{\begin{tabular}[t]{l}$r_5$\end{tabular}}}}%
    \put(0.42182826,0.19982012){\color[rgb]{0,0,0}\makebox(0,0)[lt]{\lineheight{1.25}\smash{\begin{tabular}[t]{l}$r_3$\end{tabular}}}}%
    \put(0.39031797,0.09832418){\color[rgb]{0,0,0}\makebox(0,0)[lt]{\lineheight{1.25}\smash{\begin{tabular}[t]{l}$r_2$\end{tabular}}}}%
    \put(0.53479931,0.19752493){\color[rgb]{0,0,0}\makebox(0,0)[lt]{\lineheight{1.25}\smash{\begin{tabular}[t]{l}$r_4$\end{tabular}}}}%
    \put(0.88192501,0.59094599){\color[rgb]{0,0,0}\makebox(0,0)[lt]{\lineheight{1.25}\smash{\begin{tabular}[t]{l}$r_6$\end{tabular}}}}%
    \put(0,0){\includegraphics[width=\unitlength,page=2]{grigliaEsagonale.pdf}}%
  \end{picture}%
\endgroup%
}}
   \caption{ Robot $r_6$ moving in a band during task $T_3$.}
\end{figure}

Further concepts will be introduced in the following description of the algorithm.
In the hexagonal graph, to go toward a direction, a robot either moves to the adjacent vertex if there is an edge connecting the two or it moves along the edges of the next hexagon ahead. Therefore a robot moves alternatively straight or diverting its path. As a result, the movement of a robot is enclosed in a band that is tall half the height of an hexagon while moving toward a direction. 
Given a robot and three canonical directions, there are two bands for each direction, the band selected each time by the robot is specified in a task when needed.
\begin{itemize}
\item
Task $T_1$: During this task robot $r_1$ moves away from the other robots until predicate $\vguno$ becomes true. For hexagonal grids predicate $\vguno$ is updated as follows:
\newline
$\vguno$ : $r_1$ is at a vertex such that exists a unique direction in which at least one of the lines passing through $r_1$ or one of its neighbours encounters each $\bp(R')$.

\item 
Task $T_2$: 
In this task $r_1$ moves at a distance $3\Delta$ from the origin. The origin here is redefined since it can be a vertex of $\GT$ not in $\GH$.  
Given $r_1$ and $r_n$, let $R''$ be $R''=R\setminus\{r_1,r_n\}$. Let L be the line that forms a canonical angle with $X$ passing through a robot in $R''$ and farthest from $r_1$. The origin is defined as the first vertex encountered from the intersection of $L$ and $X$ having the same type of the leading corner of $\mbp(F)$ read following the orientation of the Y-axis.

\item
Task $T_3$: In this task $r_n$ moves toward its target through any shortest path while keeping outside $mpb(R'')$ also during a detour.

\item
Task $T_4$: In this task $n-2$ robots reach their target one by one. This task develops in the same way as in $\GT$.
\item
Task $T_5$: In this task guard $r_n$ goes towards its target $f_n$. While moving parallel to the X-axis, $r_n$ moves in any band that keeps at least $2\Delta$ distance from $X$. While moving parallel to the Y-axis, $r_n$ moves in the band farthest from $r_1$.
Predicate $\vhrenne$ is updated as follows: 
\newline
$\vhrenne$ : $f_n=(x,y)$ and $r_n=(x',y')$, with $x'\leq x+1$ and $y'\geq y$   
\item
Task $T_6$: In this phase $r_1$ moves parallel to the shortest side of the parallelogram $\mbp(R)$ as to increase $d_r$.
During the movement $r_1$ moves in any band that keeps at least $3\w(\mbp(F))$ distance from $\mbp(F)$. We say that $r_1$ is in line with its target if  $d_r=d_f$ or $d_r-1=d_f$ since $r_1$ is moving within a band, so predicate $\vqfuno$ updates as follows:
\newline
$\vqfuno$ :  $\ell(\mbp(R) )  = \ell’ + \LSF$, for some $\ell’$ made of only $0$'s and just one $1$ in position
      $d_r$ and ($d_r = d_f \lor d_r-1=d_f$).
\item
Task $T_7$: In this task $r_1$ moves towards its target and in case of detours it moves in the direction such that $\ell(\mbp(R))$ decreases.
\end{itemize}

The same proofs of correctness given in Section \ref{ssec:correctness} for $\GT$ apply for $\GH$. 

\section{Conclusion}\label{sec:conclusion}
One may ask why in regular tessellation graphs $\apf$ deos not show the same solvability properties of the case of robots moving in the Euclidean plane. There, in fact, any leader configuration can be taken in input with the idea that it is always possible to break the possible symmetry by moving the leader. 
Here, in graphs, this strategy does not seem to be effective as the movements of the robots are restricted to the neighborhood. Hence, in a symmetric leader configuration it may happen the leader cannot move without causing a multiplicity which might prevent the formation of the final pattern (e.g., consider the case of a rotational configuration defined on $\GS$ with a robot on the center of rotation and all its four neighbors occupied). Hence, before moving the leader, a resolution strategy should make ``enough space'' around the leader. Actually, this approach has been followed in~\cite{C20}, a very recent work. In that paper, an algorithm able to break symmetries in leader configurations defined on $\GS$ or $\GT$ has been proposed. As a natural possible future work, it would be interesting to check whether this breaking symmetry algorithm can be composed with $\Aform$. If possible, this would completely solve the $\apf$ problem on both $\GS$ and $\GT$.

\bibliographystyle{splncs04}
\bibliography{../../../../global_references}

\begin{thebibliography}{10}
\providecommand{\url}[1]{\texttt{#1}}
\providecommand{\urlprefix}{URL }
\providecommand{\doi}[1]{https://doi.org/#1}

\bibitem{BCM16}
Bhagat, S., Chaudhuri, S.G., Mukhopadhyaya, K.: Formation of general position
  by asynchronous mobile robots under one-axis agreement. In: Proc. 10th Int.'l
  WS on Algorithms and Computation {(WALCOM)}. LNCS, vol.~9627, pp. 80--91.
  Springer (2016)

\bibitem{BAKS19}
Bose, K., Adhikary, R., Kundu, M.K., Sau, B.: Arbitrary pattern formation on
  infinite grid by asynchronous oblivious robots. In: Proc. 13th Int.'l Conf.
  on Algorithms and Computation ({WALCOM}). LNCS, vol. 11355, pp. 354--366.
  Springer (2019)

\bibitem{BKS20}
Bose, K., Adhikary, R., Kundu, M.K., Sau, B.: Arbitrary pattern formation by
  opaque fat robots with lights. In: Proc. 6th Int.'l Conf. on Algorithms and
  Discrete Applied Mathematics ({CALDAM}). LNCS, vol. 12016, pp. 347--359.
  Springer (2020)

\bibitem{BT18}
Bramas, Q., Tixeuil, S.: Arbitrary pattern formation with four robots. In:
  Proc. 20th Int.'l Symp. on Stabilization, Safety, and Security of Distributed
  Systems ({SSS}). LNCS, vol. 11201, pp. 333--348. Springer (2018)

\bibitem{C20}
Cicerone, S.: Breaking symmetries on tessellation graphs via asynchronous
  robots. In: Cordasco, G., Gargano, L., Rescigno, A. (eds.) Proceedings of the
  21st Italian Conference on Theoretical Computer Science {(ICTCS)}, 2020.
  {CEUR} Workshop Proceedings, CEUR-WS.org (2020), to appear

\bibitem{CDN18d}
Cicerone, S., {Di Stefano}, G., Navarra, A.: Gathering of robots on
  meeting-points: feasibility and optimal resolution algorithms. Distributed
  Computing  \textbf{31}(1),  1--50 (2018)

\bibitem{CDN18a}
Cicerone, S., {Di Stefano}, G., Navarra, A.: {``Semi-Asynchronous'': a new
  scheduler for robot based computing systems}. In: Proc. 38th {IEEE} Int.'l
  Conf. on Distributed Computing Systems, (ICDCS). pp. 176--187. {IEEE} (2018)

\bibitem{CDN19}
Cicerone, S., {Di Stefano}, G., Navarra, A.: Asynchronous arbitrary pattern
  formation: the effects of a rigorous approach. Distributed Computing
  \textbf{32}(2),  91--132 (2019)

\bibitem{CDN18c}
Cicerone, S., {Di Stefano}, G., Navarra, A.: Embedded pattern formation by
  asynchronous robots without chirality. Distributed Computing  \textbf{32}(4),
   291--315 (2019)

\bibitem{CDN20b}
Cicerone, S., {Di Stefano}, G., Navarra, A.: A methodology to design
  distributed algorithms for mobile entities: the pattern formation problem as
  case study. CoRR  \textbf{abs/2010.12463} (2020),
  \url{https://arxiv.org/abs/2010.12463}

\bibitem{CFPS12}
Cieliebak, M., Flocchini, P., Prencipe, G., Santoro, N.: Distributed computing
  by mobile robots: Gathering. SIAM J. on Computing  \textbf{41}(4),  829--879
  (2012)

\bibitem{DDKN12}
{D'Angelo}, G., {Di Stefano}, G., Klasing, R., Navarra, A.: Gathering of robots
  on anonymous grids and trees without multiplicity detection. Theor. Comput.
  Sci.  \textbf{610},  158--168 (2016)

\bibitem{DFPSY16}
Das, S., Flocchini, P., Prencipe, G., Santoro, N., Yamashita, M.: Autonomous
  mobile robots with lights. Theor. Comput. Sci.  \textbf{609},  171--184
  (2016)

\bibitem{DFSY15}
Das, S., Flocchini, P., Santoro, N., Yamashita, M.: Forming sequences of
  geometric patterns with oblivious mobile robots. Distributed Computing
  \textbf{28}(2),  131--145 (2015)

\bibitem{DDFN18}
{D'Emidio}, M., {Di Stefano}, G., Frigioni, D., Navarra, A.: Characterizing the
  computational power of mobile robots on graphs and implications for the
  euclidean plane. Inf. Comput.  \textbf{263},  57--74 (2018)

\bibitem{DN17}
{Di Stefano}, G., Navarra, A.: Gathering of oblivious robots on infinite grids
  with minimum traveled distance. Inf. Comput.  \textbf{254},  377--391 (2017)

\bibitem{DPV10}
Dieudonn{\'{e}}, Y., Petit, F., Villain, V.: Leader election problem versus
  pattern formation problem. In: Proc. 24th Int.'l Symp. on Distributed
  Computing (DISC). LNCS, vol.~6343, pp. 267--281. Springer (2010)

\bibitem{FPSW05}
Flocchini, P., Prencipe, G., Santoro, N., Widmayer, P.: Gathering of
  asynchronous robots with limited visibility. Theor. Comput. Sci.
  \textbf{337},  147--168 (2005)

\bibitem{FPSW08}
Flocchini, P., Prencipe, G., Santoro, N., Widmayer, P.: Arbitrary pattern
  formation by asynchronous, anonymous, oblivious robots. Theor. Comput. Sci.
  \textbf{407}(1-3),  412--447 (2008)

\bibitem{GM10}
Ghike, S., Mukhopadhyaya, K.: A distributed algorithm for pattern formation by
  autonomous robots, with no agreement on coordinate compass. In: Proc. 6th
  Int.'l Conf. on Distributed Computing and Internet Technology, ({ICDCIT}).
  LNCS, vol.~5966, pp. 157--169. Springer (2010)

\bibitem{GS87}
Gr\"unbaum, B., Shepard, G.C.: Tiling and Patterns. W. H. Freeman \& Co., New
  York (1987)

\bibitem{Ionascu12}
Ionascu, E.J.: Half domination arrangements in regular and semi-regular
  tessellation type graphs. Math  \textbf{abs/1201.4624v1} (2012),
  \url{https://arxiv.org/abs/1201.4624v1}

\bibitem{SY99}
Suzuki, I., Yamashita, M.: Distributed anonymous mobile robots: Formation of
  geometric patterns. {SIAM} J. Comput.  \textbf{28}(4),  1347--1363 (1999)

\bibitem{YS10}
Yamashita, M., Suzuki, I.: Characterizing geometric patterns formable by
  oblivious anonymous mobile robots. Theor. Comput. Sci.  \textbf{411}(26-28),
  2433--2453 (2010)

\bibitem{YUKY15}
Yamauchi, Y., Uehara, T., Kijima, S., Yamashita, M.: Plane formation by
  synchronous mobile robots in the three dimensional euclidean space. In: Proc.
  29th Int.'l Symp. on Distributed Computing {(DISC)}. LNCS, vol.~9363, pp.
  92--106. Springer (2015)

\bibitem{YY14}
Yamauchi, Y., Yamashita, M.: Randomized pattern formation algorithm for
  asynchronous oblivious mobile robots. In: Proc. 28th Int.'l Symp. on
  Distributed Computing, {(DISC)}. LNCS, vol.~8784, pp. 137--151. Springer
  (2014)

\end{thebibliography}

\end{document}